\documentclass[11pt]{article}
\usepackage{fullpage}
\usepackage[utf8]{inputenc}
\usepackage{amsmath, amsthm, amssymb, mathtools,xcolor}
\usepackage{colonequals}

\usepackage{latexsym,amsfonts,amscd,epsfig,color,mathrsfs,mdframed,mleftright}
\usepackage{aliascnt}
\usepackage{graphicx}
\usepackage{tikz}
\usepackage{algorithm}
\usepackage{algpseudocode}
\usepackage{float}
\usepackage{ifthen}
\usepackage{array}
\usepackage{pgfplots}
\usepgfplotslibrary{groupplots}
\pgfplotsset{compat=1.18}
\usetikzlibrary{decorations.pathreplacing}
\usepackage{comment}

\usepackage{hyperref}
\hypersetup{
    colorlinks=true,
    allcolors=blue
}
\usepackage[capitalize,nameinlink]{cleveref}
\crefname{claim}{claim}{claims}
\Crefname{claim}{Claim}{Claims}
\crefname{question}{question}{questions}
\Crefname{question}{Question}{Questions}
\crefname{conjecture}{conjecture}{conjectures}
\Crefname{conjecture}{Conjecture}{Conjectures}
\crefname{observation}{observation}{observations}
\Crefname{observation}{Observation}{Observations}

\newcommand{\ket}[1]{|#1\rangle}
\newcommand{\bra}[1]{\langle#1|}
\newcommand{\ketbra}[2]{|#1\rangle\langle#2|}
\newcommand{\braket}[2]{\langle#1|#2\rangle}
\newcommand{\norm}[1]{\left\lVert#1\right\rVert}
\newcommand{\ceil}[1]{\lceil{#1}\rceil}
\newcommand{\floor}[1]{\lfloor{#1}\rfloor}

\newtheorem{theorem}{Theorem}[section]
\newtheorem{corollary}[theorem]{Corollary}
\newtheorem{remark}[theorem]{Remark}
\newtheorem{lemma}[theorem]{Lemma}
\newtheorem{claim}[theorem]{Claim}

\newtheorem{definition}[theorem]{Definition}

\def\01{\{0,1\}}
\newcommand{\C}{\mathbb{C}}
\newcommand{\eps}{\varepsilon}
\renewcommand{\epsilon}{\varepsilon}

\newcommand{\mathify}[1]{\ifmmode{#1}\else\mbox{$#1$}\fi}

\newcommand{\ot}{\otimes}
\def\abs#1{\left| #1 \right|}

\renewcommand{\H}{\mathcal{H}}

\renewcommand{\L}{\mathcal{L}}
\newcommand{\R}{\mathbb{R}}
\newcommand{\Z}{\mathbb{Z}}
\newcommand{\N}{\mathbb{N}}
\newcommand{\E}{\mathbb{E}}

\newcommand{\defeq}{\colonequals}

\title{Optimal Ground-State Preparation 
with a Guiding State}
\author{Stacey Jeffery\thanks{QLever, QuSoft, CWI and University of Amsterdam, the Netherlands. This work is co-funded by the European Union (ERC, ASC-Q, 101040624).}
\and
Rolando D. Somma\thanks{Google Quantum AI, Venice, CA 90291, United States.}
\and
Freek Witteveen\thanks{QuSoft and CWI, Amsterdam, the Netherlands}
\and 
Ronald de Wolf\thanks{Google Quantum AI, Venice, CA 90291, United States. Also QuSoft, CWI and University of Amsterdam, the Netherlands. Partially supported by the Dutch Research Council (NWO) through Gravitation-grant Quantum Software Consortium, 024.003.037.}
}
\date{}

\begin{document}

\maketitle

\begin{abstract}
Suppose a Hamiltonian $H$ has a unique ground state $\ket{\psi_0}$ with an eigenvalue $E_0$, and we have an estimate $\tilde{E}_0$ such that $|\tilde{E}_0-E_0|\leq\delta$, and there is a gap of at least $3\delta$ between $E_0$ and all other eigenvalues. Suppose we have a unitary $A$ available that can produce a ``guiding state'' $A\ket{0}$ that has overlap at least $\gamma$ with $\ket{\psi_0}$. We show how to obtain an $\eta$-approximation of $\ket{\psi_0}$ with probability at least $1-\eps$ using $O(\log(1/\eps)/\gamma\delta + \log(1/\eta)/\delta)$ applications of $U=e^{iH}$ its inverse, and  $O(\log(1/\eps)/\gamma)$ applications of $A$ and its inverse.
We give two different algorithms, one based on interleaving amplitude amplification and error-reduction in the style of~\cite{hmw:berrorsearch}, and one using the composition of transducers.
This paper is the state-preparation follow-up to our two recent ground-state-energy estimation papers~\cite{JW:optQPE,SdW:optQPE}. Combined with the optimal energy-estimation of~\cite{JW:optQPE}, our results show  the claimed upper bound for ground-state preparation. The bound is optimal up to a constant factor in terms of the number of applications of $U$ and $U^\dagger$ if $\eps=\eta$, as follows from~\cite{SdW:optQPE}.

In one of our approaches to ground state preparation, we construct transducers for amplitude amplification and LCU that might be of independent interest. 
\end{abstract}



\section{Introduction}

We study the problem of preparing the unique ground state $\ket{\psi_0}$ for a given Hamiltonian~$H$. The tools we have available for this are the unitary $e^{iH}$ (which one gets from doing Hamiltonian simulation on $H$) and a unitary $A$ that prepares a ``guiding state'' $A\ket{0}$ that is promised to have overlap at least $\gamma$ with $\ket{\psi_0}$. One may think of $A\ket{0}$ as an ``ansatz'' for the ground state, based on physical or chemical intuition about the system described by $H$.

A closely related problem is estimating the ground-state \emph{energy} $E_0$ associated with the ground state $\ket{\psi_0}$ (i.e., $H\ket{\psi_0}=E_0\ket{\psi_0}$).
For the purposes of normalization and to avoid the problem that eigenvalues 0 and $2\pi$ are indistinguishable for $U=e^{iH}$, let us assume all eigenvalues are in the interval $[0,\pi]$.
This problem of estimating~$E_0$ through a variant of phase estimation is extensively studied in the quantum computing literature. Recently, two of us~\cite{JW:optQPE} addressed this problem: one can estimate $E_0$ to within $\pm\delta$ using $O(1/\gamma\delta)$ applications of $U$ and $U^\dagger$. This improves previous work by a factor of $\log(1/\gamma)$, and matches an earlier lower bound of~\cite{Mande2026tightboundsquantum} up to a constant factor. At the expense of a factor of $\log(1/\eps)$, one can then reduce the error probability to small~$\eps$, and this upper bound of $O(\log(1/\eps)/\gamma\delta)$ applications of $U$ and $U^\dagger$ was very recently shown to be optimal as well by the other two of us~\cite{SdW:optQPE}. 

In this paper we focus on efficiently preparing a state that is $\eta$-close to the ground state $\ket{\psi_0}$, with probability of success $1 - \eps$.
To make the problem tractable, we also assume the Hamiltonian is gapped: all other eigenvalues are at least $E_0+\Delta$.
We may assume we have already obtained an estimate of the ground state energy $E_0$ to within additive error $\delta \leq  \Delta/3$, for instance by running the algorithm from~\cite{JW:optQPE}. 
\cite{SdW:optQPE} already proved a lower bound of $\Omega(\log(1/\eps)/\gamma\Delta)$ applications of $U$ and $U^\dagger$ for this state-preparation problem when $\eps=\Theta(\eta)$.

First consider the case of constant $\eps$ and $\eta$, which can then be relatively easily extended to arbitrarily small $\eps$ and $\eta$.
Here we give two different algorithms that both prepare an approximation of $\ket{\psi_0}$ using the optimal $O(1/\gamma\Delta)$ applications of $U$ and $U^\dagger$.
Both follow the high-level idea of using the estimate of $E_0$ to ``mark'' the ground space, and then use amplitude amplification to boost it. Naively, one could mark the ground space by using phase estimation to precision $\delta/4$, and comparing the estimated phase with $\tilde{E}_0$, however, this has nonzero probability, and naive success probability amplification incurs an overhead of $O(\log\frac{1}{\gamma})$, since this subroutine will be called $O(1/\gamma)$ times by amplitude amplification. This gives a total complexity of $O(\frac{1}{\gamma\Delta}\log\frac{1}{\gamma})$, which is worse than the best known complexity for estimating the ground-state energy by a factor $O(\log\frac{1}{\gamma})$. Our two algorithms deal with the issue of the subroutine error in two different ways, each avoiding the $O(\log\frac{1}{\gamma})$ factor in their own way.

Our first algorithm (\Cref{sec:HMWbasedproof}) uses fairly well-established techniques, specifically the ``quantum search on bounded-error inputs'' of~\cite{hmw:berrorsearch}. That paper considered a search problem on $N$ bits $x_0,\ldots,x_{N-1}$, where for each $j$ we have a coherent quantum algorithm that computes $x_j$ with error probability $\leq 1/3$: it marks the basis states that are solutions, albeit with some error. Suppose $x_i$ is the only one of these bits that is~1, and we want to find this $i$. Suppose we also have a unitary $A$ such that $A\ket{0}$ has overlap $\geq\gamma$ with the unique solution~$\ket{i}$. \cite{hmw:berrorsearch} shows how to do amplitude amplification to increase that overlap to a constant (at which point a measurement will yield the solution $i$ with constant probability). It does so by a subtle recursive amplitude amplification interleaved with error-reduction. Each round of amplitude amplification increases the magnitude of the marked part of the state by a factor of roughly~3. This amplifies the solution~$\ket{i}$, but it also amplifies the part of the state that was incorrectly marked as a solution by the bounded-error queries.
By doing the right amount of error-reduction in between the amplification steps in order to push these ``false positives'' down again, this method still works in asymptotically the same number of queries as in the noiseless case.
The problem we solve here is analogous, with $N$ being the dimension of $H$, and $\ket{\psi_0}$ taking the role of the solution $\ket{i}$.  
We can now use phase estimation with precision $\delta$ to mark the solution (i.e., to distinguish eigenphase $E_0$ from eigenphases $\geq E_0+\Delta=E_0+3\delta$), albeit with some error.
The difference with~\cite{hmw:berrorsearch} is that we are now working in the unknown basis of eigenstates of~$H$ rather than in the known computational basis. Despite this complication, we show that their technique can be made to work for our ground-state preparation problem as well, yielding the optimal complexity.

Our second algorithm (\Cref{sec:transducerbasedproof}) is based on the more modern ``transducer'' techniques (see \Cref{sec:transducers}) that were already used in~\cite{JW:optQPE} for estimating $E_0$. A transducer is a way of representing a bounded-error quantum algorithm that composes very nicely. One reason for this nice composition is that transducers can often be exact, even when they represent bounded-error algorithms, so they compose without error. When the final composed transducer is mapped back to a standard quantum algorithm, the algorithm will have error, but we can often avoid log factors from error reduction by using transducers. The transducer for preparing the ground state is a composition of two transducers. The first transducer, presented in \Cref{sec:lcu-transducer}, is for LCU (linear combination of unitaries) where the unitaries are all powers of some $U$. A special case of this LCU applies a filter that marks only the ground state. The second transducer, presented in \Cref{sec:aa-transducers}, is for amplitude amplification, in order to amplify the marked ground state. These compose without any need for error reduction, to obtain a transducer with ``complexity'' $O(1/(\gamma\Delta))$, which can then be turned into a bounded-error quantum algorithm with this optimal complexity. The transducers for LCU and amplitude amplification may be of independent interest. A previous transducer for amplitude amplification was used in~\cite{apers2026elfs}, but our construction improves the space complexity by a multiplicative log factor.

In both algorithms, we can reduce the error parameters $\eps$ and $\eta$ from constant to an arbitrary value by using the ability to check if we are in the ground state (see \Cref{ssec:HMdWalgoreduc}).

\subsection{Problem setup, and formal statement of results}\label{sec:setup}
For a Hamiltonian evolution unitary $U$ on $\mathbb{C}^N$, we write
$$U=\sum_{k=0}^{K-1}e^{i E_k}\Pi_k, \qquad 0 \leq E_0<E_1<\dots<E_{K-1}\leq \pi.$$
The lower bound on $E_0$ is just for convenience: it only matters that modulo $2\pi$, $E_0$ should be bounded away from $E_{K-1}$.
We will also assume access to an \emph{advice-preparation unitary} $A$, with
$$A\ket{0}=\ket{\psi}=\sum_{k=0}^{K-1}\alpha_k\ket{\psi_k},$$ 
for some amplitudes $\alpha_0,\dots,\alpha_{K-1}$ and normalized states $\ket{\psi_0},\dots,\ket{\psi_{K-1}}$ such that for all $k$:
$$\alpha_k\in \R_{\geq 0},\;\mbox{ and }\; \Pi_k\ket{\psi_k}=\ket{\psi_k}.$$
We will assume that the ground state is unique, meaning $\Pi_0=\ket{\psi_0}\bra{\psi_0}$ has rank one. 

\begin{definition}[Guided ground-state preparation]
Fix $\Delta\in (0,\pi]$, $\gamma\in (0,1]$, and $\epsilon,\eta\in (0,1/2)$.
    Given query access to $(U,A)$ as above, such that $U$ has a unique ground state, $E_1\geq E_0+\Delta$, and $\alpha_0\geq \gamma$, output a state $\rho$ such that $\norm{\ketbra{\psi_0}{\psi_0} - \rho}_1 \leq \eta$ with probability of success at least $1 - \eps$.  
\end{definition}

\begin{definition}[Guided ground-state preparation, with estimate]\label{def:groundstateprep} 
Fix $\Delta\in (0,\pi]$, $\gamma\in (0,1]$, and $\epsilon,\eta\in (0,1/2)$.
    Given query access to $(U,A)$ as above, where $U$ has a unique ground state, $E_1\geq E_0+\Delta$, and $\alpha_0\geq \gamma$; and an estimate $\tilde{E}_0$ such that $|\tilde{E}_0-E_0|\leq \delta\defeq \Delta/3$, output a state $\rho$ such that $\norm{\ketbra{\psi_0}{\psi_0} - \rho}_1 \leq \eta$ with probability of success at least $1 - \eps$.  
\end{definition}

The main result of this work is the following.

\begin{theorem}\label{thm:main}
There is a quantum algorithm that solves the problem of  guided ground-state preparation with estimate using
$O(\frac{1}{\gamma \Delta}\log \frac{1}{\eps} +\frac{1}{\Delta} \log\frac{1}{\eta})$ controlled calls to $U$ and $U^\dagger$ and $O(\frac{1}{\gamma}\log \frac{1}{\eps})$ controlled calls to $A$ and $A^\dagger$.
\end{theorem}

We prove this theorem in two different ways, using the two methods described above. Combining this theorem with our previous algorithm for estimating the ground energy~\cite[Theorem~1.1]{JW:optQPE}, we immediately get the following corollary.

\begin{corollary}
There is a quantum algorithm that solves the guided ground-state preparation problem using 
$O(\frac{1}{\gamma \Delta}\log \frac{1}{\eps} +\frac{1}{\Delta} \log\frac{1}{\eta})$ controlled calls to $U$ and $U^\dagger$ and $O(\frac{1}{\gamma}(\log\frac{1}{\Delta}\log\log\frac{1}{\Delta} + \log \frac{1}{\eps}))$ controlled calls to $A$ and $A^\dagger$.
\end{corollary}

We have thus improved the number of queries to $U$, which is typically the dominating factor in the complexity, by a factor of $O(\log\frac{1}{\gamma})$ relative to the previous best known algorithm. By \cite{SdW:optQPE}, this is optimal.
For the approach using transducers, we also bound the number of additional one- and two-qubit gates, and the number of ancilla qubits. We find that a logarithmic number of ancilla qubits, and a number of gates that scales as the query complexity in \cref{thm:main} up to logarithmic factors, suffice. See \cref{thm:gsp-constant} for details.

\paragraph{Concurrent work.} An independent work~\cite{chen2026optimal} similarly builds on~\cite{JW:optQPE}, using transducers to get an algorithm for ground-state preparation which is similarly optimal in the number of queries to $U$ and $A$.

\paragraph{AI statement.}
The majority of high-level ideas in this work are human.
The algorithm in \cref{ssec:HMdWalgo} is fully devised by the authors, though the error-reduction in \cref{ssec:HMdWalgoreduc} got some help from LLMs (ChatGPT-6-Astra).
Some important ideas in \Cref{sec:transducerbasedproof} come from LLMs, namely the leak schedule in the amplitude amplification transducer (while the overall structure of that transducer was a human idea); and the choice of filter function in \Cref{sec:filter}, as well as the method for preparing its Fourier coefficients. 

We also used LLMs for generating some figures and proofs, and extensively for proofreading. All proofs and text have been checked carefully and rewritten as needed, and the authors take full responsibility for the correctness of all parts of the paper.

\section{An algorithm using robust amplitude amplification~\cite{hmw:berrorsearch}}\label{sec:HMWbasedproof}

\subsection{Increasing the overlap with the ground state from $\gamma$ to a constant}\label{ssec:HMdWalgo}

We want to use a version of amplitude amplification to grow the overlap of $A\ket{0}$ with the ground state $\ket{\psi_0}$ from $\gamma$ to a constant.
To that end we need a unitary that ``marks'' the unique ground state, in the sense of setting a flag qubit to~1 if the first register contains $\ket{\psi_0}$, and keeping the flag at 0 if the first register contains any excited states, i.e., any state orthogonal to $\ket{\psi_0}$. We cannot do this marking perfectly, but we can do it with small, tunable error~$\xi$ using phase estimation if we already have a good estimate of $E_0$:

\begin{lemma}\label{lem:PEverifier}
Suppose all eigenvalues of $H$ are in $[0,\pi]$ and we have an estimate $\tilde{E}_0$ such that $|\tilde{E}_0-E_0|\leq\delta$, and all other eigenvalues are at least $E_0+\Delta$ for $\Delta=3\delta$. 
For each $\xi\in(0,1/2]$, there exists a unitary $Q_\xi$ involving $O(\log(1/\xi)/\delta)$ controlled applications of $U$ and $U^\dagger$, and with $m=O(\log(1/\xi)\log(1/\delta))$ auxiliary qubits that are initially~$\ket{0}$ (and are mostly $\ket{0}$ again at the end), such that for all reals $a,b$ with $a^2+b^2=1$,  state $\ket{w_0}$, and state $\ket{\psi_0^\perp}$ that has no support on $\ket{\psi_0}$ in its first register, we have
\[
\norm{Q_\xi(a\ket{\psi_0}\ket{w_0}\ket{0^{m+1}}+b\ket{\psi_0^\perp}\ket{0^{m+1}})
-
(a\ket{\psi_0}\ket{w_0}\ket{0^m}\ket{1}+b\ket{\psi_0^\perp}\ket{0^m}\ket{0})
}\leq\xi.
\]
This algorithm preserves the eigenstates of~$H$.
\end{lemma}

\begin{proof}
We just sketch the idea, using standard techniques.
Use phase estimation with precision $\delta/4$, $r=O(\log(1/\xi))$ times on the state. Apply an $X$-gate on the last qubit conditioned on having more than half of the $r$ estimates being less than $\tilde{E}_0+1.5\delta$. By a Chernoff bound, the error can be made  $\ll\xi^2$ by choosing a sufficiently large constant in the $O(\cdot)$ of~$r$. Now reverse the $r$ phase estimations to set the $m$ auxiliary qubits back to $\ket{0}$.
\end{proof}

The number of auxiliary qubits here could be reduced by using QSVT to implement $Q_\xi$.
Since we will vary $\xi$ in different applications of $Q_\xi$, the number of auxiliary qubits~$m$ will also vary. We will use a fresh batch of auxiliary qubits in each application of $Q_\xi$, absorbing those qubits into the workspace instead. Below we will omit writing those auxiliary qubits for simplicity.

For error bounds $\eta_1,\eta_2,\ldots$ that we will choose later, we recursively define algorithms on a growing number of qubits as follows:
\[
A_1=Q_{\eta_1}A
\mbox{~~~and~~~}
A_{\ell+1}=Q^c_{\eta_{\ell+1}}(
(\underbrace{(A_\ell
(2\ketbra{0}{0}-I)
A_\ell^{-1})
(I\otimes Z)}_{{\rm amplitude~amplification~step}~B}A_\ell)\otimes I_2)
\]
Each of the algorithms $A_2,A_3,\ldots$ adds the fresh auxiliary $\ket{0}$-qubits that it needs for its application of $Q_{\eta_j}$, and one more $\ket{0}$-qubit  that will be used as the new flag qubit by its application of $Q_{\eta}$.
The $Q_{\eta_j}$ will be controlled by the previous last qubit, which was the flag qubit of $A_\ell$. The superscript $c$ indicates this control.

Initially we have the state
\[
A\ket{0}=a_0\ket{\psi_0}\ket{w_0}+\sqrt{1-a_0^2}\ket{\psi_0^\perp},
\]
where the unknown overlap $a_0>0$ is at least the known value $\gamma$, $\ket{w_0}$ is the workspace of the guiding-state-generating algorithm, and $\ket{\psi_0^\perp}$ is some two-register state that has no support on $\ket{\psi_0}$ in its first register.
$A_1$ applies $Q_{\eta_1}$ with an extra $\ket{0}$-qubit (which is intended to  flag whether the first register contains $\ket{\psi_0}$), obtaining the state
\begin{equation}\label{eq:A10decomposition}
A_1\ket{0}\ket{0} =Q_{\eta_1}(A\ket{0})\ket{0}=a_1\ket{\psi_0}\ket{w'_0}\ket{1}+b_1\ket{(\psi_0^\perp)'}\ket{1}+\sqrt{1-|a_1|^2-|b_1|^2}\ket{\chi_1}\ket{0}.
\end{equation}
where $a_1,b_1$ are nonnegative reals satisfying $a_1\geq a_0(1-\eta_1)\geq\gamma(1-\eta_1)$ and $b_1\leq \eta_1$, and $\ket{w'_0}$, $\ket{(\psi_0^\perp)'}$, and $\ket{\chi_1}$ are normalized states.

The goal of these algorithms is to increase the weight of $\ket{\psi_0}$ in the first register, and to mark it with a 1 in the newly added flag-qubit at the end of the state.
Note that $A_\ell$ involves 1 application of $Q_{\eta_\ell}$, 3 applications of $Q_{\eta_{\ell-1}},\ldots,3^{\ell-1}$ applications of $Q_{\eta_1}$. 
Since $Q_\eta$ uses $O(\log(1/\eta)/\delta)$ applications of $U$ and $U^\dagger$, the total number of applications if we recurse $L$ times is 
\begin{equation}\label{eq:cost of tildeA}
\sum_{\ell=1}^L  3^{L-\ell}O(\log(1/\eta_\ell)/\delta).
\end{equation}
We now want to choose $L$ and $\eta_1,\ldots,\eta_L$ in such a way that the first register of $A_L\ket{0}$ has $\Omega(1)$ overlap with $\ket{\psi_0}$, while keeping the complexity of Equation~\eqref{eq:cost of tildeA} to be $O(1/\gamma\delta)$. 

Let's analyze what happens when going from $A_\ell$ to $A_{\ell+1}$.
We split the state $A_\ell\ket{0}$ into three pieces like we did for $A_1\ket{0}$ in \cref{eq:A10decomposition}: the part that has $\ket{\psi_0}$ in the first register and where the last qubit (the flag) is~1, the part where the first register is orthogonal to $\ket{\psi_0}$ and yet the last qubit is~1 (the ``false positive'' part of the state), and the part where the last qubit is~0.
For some nonnegative reals $a_\ell,b_\ell$, and normalized states $\ket{w_0}$, $\ket{\psi_0^\perp}$ (whose first register has no support on 
$\ket{\psi_0}$), and $\ket{\chi_\ell}$, we can write 
\[
A_\ell\ket{0}=a_\ell\ket{\psi_0}\ket{w_0}\ket{1}+b_\ell\ket{\psi_0^\perp}\ket{1}+\sqrt{1-a_\ell^2-b_\ell^2}\ket{\chi_\ell}\ket{0}.
\]
Define angle $\theta_\ell\in[0,\pi/2]$ such that $\sin(\theta_\ell)^2=a_\ell^2+b_\ell^2$, which is the squared norm of the part of the state ending in $\ket{1}$. Then $\cos(\theta_\ell)^2=1-a_\ell^2-b_\ell^2$ is the squared norm of the part ending in~$\ket{0}$.

$B=(A_\ell
(2\ketbra{0}{0}-I)
A_\ell^{-1})\cdot (I\otimes Z)$ is the product of a reflection about the subspace of states ending with $\ket{1}$ followed by a reflection about $A_\ell\ket{0}$. 
This is one perfect amplitude amplification step, which amplifies the part of the state ending in~$\ket{1}$, increasing the angle from $\theta_\ell$ to $3\theta_\ell$. We obtain 
\[
BA_\ell\ket{0}=\frac{\sin(3\theta_\ell)}{\sin(\theta_\ell)}a_\ell\ket{\psi_0}\ket{w_0}\ket{1}+\frac{\sin(3\theta_\ell)}{\sin(\theta_\ell)}b_\ell\ket{\psi_0^\perp}\ket{1}+\cos(3\theta_\ell)\ket{\chi_\ell}\ket{0}.
\]
This step has increased the part of the state where the flag-qubit is~1 by a factor $\sin(3\theta_\ell)/\sin(\theta_\ell)$, which is approximately~3 as long as $\theta_\ell$ is small. This step amplifies both the part of the state with $\ket{\psi_0}$ in the first register, as intended, but equally amplifies the ``false positive'' state $\ket{\psi_0^\perp}\ket{1}$.

We now want to use $Q_{\eta_{\ell+1}}$ to ``push down'' the weight of this false positive state.
$A_{\ell+1}=Q^c_{\eta_{\ell+1}}(B\otimes I)(A_\ell\otimes I)$ applies $Q_{\eta_{\ell+1}}$ on the above state, conditioned on its last qubit, with an extra $\ket{0}$-qubit that should flag whether the first register is $\ket{\psi_0}$. 
This slightly reduces the amplitude on the part of the state that has $\ket{\psi_0}$ in the first register and the final qubit set to~1, but it greatly decreases the amplitude on $\ket{\psi_0^\perp}\ket{1}$. Hence this step greatly increases the relative weight of the state with $\ket{\psi_0}$ within the part of the state where the newly introduced flag-qubit is $\ket{1}$.
Write the new state (which has one qubit more than $A_\ell\ket{0}$) as
\[
A_{\ell+1}\ket{0}
=a_{\ell+1}\ket{\psi_0}\ket{w'_0}\ket{1}+b_{\ell+1}\ket{(\psi_0^\perp)'}\ket{1}+\sqrt{1-a_{\ell+1}^2-b_{\ell+1}^2}\ket{\chi_{\ell+1}}\ket{0}.
\]
We now want to show that $a_\ell$ grows by a factor nearly~3 in each recursive step, and that the ratio between $b_\ell$ and $a_\ell$ goes down significantly (the latter means that more and more of the flagged substate has the actual ground state $\ket{\psi_0}$ in its first register).
First, using the trigonometric identity $\sin(3\theta)/\sin(\theta)=3-4\sin(\theta)^2$, we have
\[
a_{\ell+1}\geq \frac{\sin(3\theta_\ell)}{\sin(\theta_\ell)}a_\ell(1-\eta_{\ell+1})\geq a_\ell(3-4(a_\ell^2+b_\ell^2))(1-\eta_{\ell+1}).
\]
Unfolding these inequalities, we have
\[
a_L\geq 3^{L-1} a_1(1-\frac{4}{3}\sum_{\ell=1}^{L-1}(a_\ell^2+b_\ell^2)-\sum_{\ell=1}^{L-1}\eta_{\ell+1}).
\]
Second, because $B$ keeps the ratio between $b_\ell$ and $a_\ell$ the same but $Q_{\eta_{\ell+1}}$ shrinks the former by at least a factor $\eta_{\ell+1}$ but the latter by at most $1-\eta_{\ell+1}$, we have
\[
\frac{b_{\ell+1}}{a_{\ell+1}}\leq \frac{b_\ell}{a_\ell}\cdot\frac{\eta_{\ell+1}}{1-\eta_{\ell+1}}.
\]
Unfolding this, using that we will choose $\eta_\ell$ to be $\leq 1/2$ for all $\ell$, and hence $1/(1-\eta_{\ell+1})\leq 2$, gives
\[
\frac{b_{L}}{a_{L}}\leq\frac{b_1}{a_1}\prod_{\ell=1}^{L-1} 2\eta_{\ell+1}.
\]
Recall that $a_1\geq\gamma(1-\eta_1)\geq \gamma/2$ and $b_1\leq\eta_1$, hence $b_1/a_1\leq 2\eta_1/\gamma$.

We will choose 
\[
L=\floor{\log_3(1/\gamma)}\mbox{ and }\eta_{\ell}=2^{-\ell}/100\mbox{ for }\ell=1,\ldots,L.
\]
Note that because $\sum_{\ell\geq 1}\ell\cdot 3^{-\ell}=O(1)$, the complexity of \cref{eq:cost of tildeA} is then $O(3^L/\delta)=O(1/\gamma\delta)$, as promised. 

With these choices, we will very quickly shrink $b_\ell$ relative to $a_\ell$:
\[
b_\ell/a_\ell\leq (b_1/a_1)(1/100)^{\ell-1}2^{-\ell(\ell-1)/2}\mbox{ and }b_1/a_1\leq 2\eta_1/\gamma.
\]
We also have $\sum_{\ell=1}^{L-1}\eta_{\ell+1}\leq 1/100$. Working this out, it follows that if the initial overlap $a_0$ is not too much bigger than its known lower bound~$\gamma$ (say, $a_0\in[\gamma,3\gamma]$), then $A_L\ket{0}$ will have $a_L=\Omega(1)$, i.e., it will have an $\Omega(1)$ weight on the subspace where the first register is the ground state $\ket{\psi_0}$ and the flag-qubit is~1.

Because $a_0$ might, however, be much larger than the lower bound $\gamma$, we can actually try out all recursion depths $\ell\in\{1,\ldots,L\}$. Since their costs scale like $3^\ell/\delta$, trying all such~$\ell$ is asymptotically not more expensive than only trying $L$. One of those runs (namely where $\ell=\Theta(a_0 3^\ell)$) will produce a state with an $\Omega(1)$ weight on the subspace where the first register is the ground state~$\ket{\psi_0}$.
Measuring the flag qubit for each of these $L$ states gives, with constant success probability, at least one state that has constant overlap with $\ket{\psi_0}$, and we will choose the one with largest~$\ell$.%
\footnote{Note that a flag-qubit measurement outcome of~1 for any $\ell\geq\sqrt{2\log_2(1/a_0)}$ 
guarantees that the first register of the corresponding post-measurement state has large overlap with $\ket{\psi_0}$, because $b_\ell/a_\ell$ is upper bounded by a small constant even already for $\ell=\sqrt{2\log_2(1/a_0)}$ (and it only shrinks further for larger $\ell$).} 
We can increase the success probability to any constant close to~1 by $O(1)$ repetitions. 

Note that we can also treat such algorithms coherently: instead of actually measuring the flag-qubits, we can condition our later operations (such as which register from the $L$ recursion depths to swap into the first register at the end) on their value. That way we can construct a unitary ${\cal A}$ of cost $O(1/\gamma\delta)$ such that ${\cal A}\ket{0}=a\ket{\psi_0}\ket{w_0}+\sqrt{1-a^2}\ket{\psi_0^\perp}$ for some $a>1/10$, some arbitrary work-space state $\ket{w_0}$, and arbitrary state $\ket{\psi_0^\perp}$ that is orthogonal to $\ket{\psi_0}$ in its first register (i.e., that only has support on excited states).

\subsection{Reducing the trace distance~$\eta$ from $\ket{\psi_0}$, and the error probability~$\eps$}\label{ssec:HMdWalgoreduc}

Suppose algorithm $\mathcal A$ (for instance the algorithm from the end of \cref{ssec:HMdWalgo}, or the transducer-based algorithm of the next section) prepares a state 
\[
{\cal A}\ket{0}=a\ket{\psi_0}\ket{w_0}+\sqrt{1-a^2}\ket{\psi_0^\perp}\] 
that has 
$a^2\geq 1/100$, at cost $c_{\mathcal A}=O(1/(\gamma\delta))$.
As before ``cost'' counts the total number of applications of $U$ and $U^\dagger$, including controlled
applications. Let $Q_\xi$ be a phase-estimation-based verifier (from \cref{lem:PEverifier}) of cost
$O(\log(1/\xi)/\delta)$ that preserves the ground space and excited subspaces, and that sets a flag-qubit with false-rejection and false-acceptance probabilities at most
$\xi^2$ on these respective subspaces. Each run of this verifier below uses fresh auxiliary qubits.

Fix target parameters $\eta,\eps\in(0,1/2)$ and set
\[
J=\left\lceil\log_2\log_2(32/\eta)\right\rceil,\qquad\eta_j=2^{-2^j}\mbox{ for }j=1,\ldots,J.
\]
Consider the following algorithm, which we call one ``attempt'':

\medskip

\noindent
\hspace*{2em}Prepare $\mathcal A\ket0$.\\
\hspace*{2em}Successively for $j=1,\ldots,J$ apply
$Q_{\eta_j}$ and measure its new flag-qubit.\\ 
\hspace*{2em}At the first flag-measurement outcome $0$,
discard the state and reject the attempt.\\ \hspace*{2em}If all flag-outcomes are $1$, then accept and output
the first register of the final state.

\medskip

\noindent
The unnormalized  weight on the ground state in the accepting branch is at least
\[
 a^2\prod_{j=1}^J(1-\eta_j^2)
 \ge\frac1{100}\left(1-\sum_{j=1}^J\eta_j^2\right)
 >\frac1{200}=:p_*.
\]
The unnormalized weight on excited states in the accepting branch is at most
$(1-a^2)\prod_{j=1}^J\eta_j^2\le\eta_J^2$.
Hence the state $\rho$ of the first register of every accepted state satisfies
\begin{equation}\label{eq:closetopsio}
 \|\ketbra{\psi_0}{\psi_0}-\rho\|_1
 \le2\sqrt{1-\langle\psi_0|\rho|\psi_0\rangle}
 \le2\sqrt{200}\,\eta_J<\eta.
\end{equation}
Let $p_s$ be the probability of rejecting at verification-stage
$s$, meaning that the $s$th flag-qubit measurement is the first to give outcome~0. Note that $p_1,\ldots,p_J$ together with the probability of acceptance form a probability distribution (they correspond to disjoint events and sum to~1).
Since the ground and excited states are preserved by $Q_\xi$, and the probability for the ground state to be rejected at stage~$s$ is $\leq \eta_s^2$, and the probability for an excited state to be accepted at stage $j<s$ is $\leq \eta_j^2$, we have
\[
 p_s\le a^2\eta_s^2+(1-a^2)\prod_{j<s}\eta_j^2
 =\bigl(a^2+16(1-a^2)\bigr)\eta_s^2
 \leq 16\cdot2^{-2^{s+1}}.
\]
This very rapid decay of the probabilities $p_s$ with $s$ means that conditioned on the event that an attempt is rejected, it is extremely likely to end very early in the process, i.e., for a small~$s$. The verification costs
up to that point are relatively low, since  the cost of $Q_{\eta_j}$ is $O(\log(1/\eta_j)/\delta)=O(2^j/\delta)$. This is why we do the verification in $J$ separate steps of increasing cost, rather than always doing only one expensive step that immediately pushes the trace distance w.r.t.\ $\ket{\psi_0}$ down to $\eta$. This is similar to but independent of how it's done in~\cite[Section~7.2]{chen2026optimal}.

Choose an integer $d=O(1/\delta)$ so that stage $j$ costs at most
$d2^j$. A rejection at stage $s$ incurs verification cost at most
$\sum_{j=1}^s d2^j=d(2^{s+1}-2)\le2d\,2^s$. Define random variable $X$ as this verification cost on
rejection and as zero on acceptance. For the purposes of showing a tail bound on sums of independent $X$s later, we upper bound the exponential moment of $X$ as follows
\[
 \E[2^{X/(2d)}]
 \le1+\sum_{s=1}^Jp_s2^{2^s}
 \le1+16\sum_{s\ge1}2^{-2^s}
 \le1+\frac{16}{3}<7.
\]
The overall algorithm will use at most $R:=\lceil200\ln(2/\eps)\rceil$ independent
attempts, stopping at the first acceptance. 

The probability that all $R$ attempts reject is at most $(1-p_*)^R\le\eps/2$. For $R$ hypothetical
complete attempts, let $X_1,\ldots,X_R$ be the corresponding independent random 
variables for their respective verification costs. Setting
\[
 T=2d\bigl(R\log_2 7+\log_2(2/\eps)\bigr),
\]
Markov's inequality and independence give the following tail bound
\begin{align*}
\Pr\!\left[\sum_{i=1}^R X_i>T\right]
 & =\Pr\!\left[2^{\sum_{i=1}^R X_i/(2d)}>2^{T/(2d)}\right]\\
& \leq\frac{\E[2^{\sum_{i=1}^R X_i/(2d)}]}{2^{T/(2d)}}
 =
\frac{\prod_{i=1}^R \E[2^{ X_i/(2d)}]}{2^{T/(2d)}}
 \leq
 \frac{7^R}{2^{T/(2d)}}=\eps/2.
\end{align*}
There is a small probability that the verification costs summed over all the $R$ attempts are much higher than expected and together exceed the upper bound on the total cost that we are aiming at.
In order to get a worst-case complexity, we will abort when the total verification cost of rejected attempts exceeds~$T$. By the previous inequality, this happens with  probability $\leq\eps/2$. Hence by a union bound, the probability that no output state is produced is $\leq \eps/2+\eps/2=\eps$. 
If an output state is produced, it is $\eta$-close to $\ketbra{\psi_0}{\psi_0}$ with certainty by \cref{eq:closetopsio}, so the algorithm indeed achieves what \cref{def:groundstateprep} requires.

It remains to tally up the costs of the algorithm.
Preparing ${\cal A}\ket{0}$ in each of the $\leq R$ attempts contributes $\leq Rc_{\cal A}$ to the total cost.
There is at most one accepted attempt, whose verification cost is at
most $F=d(2^{J+1}-2)=O(\log(1/\eta)/\delta)$. 
The sum of the verification costs of all the rejected attempts together is capped at~$T$.
Hence the total worst-case cost is 
\[
\displaystyle
 B=Rc_{\mathcal A}+F+T=O\!\left(
 \frac{\log(1/\eps)}{\gamma\delta}
 +\frac{\log(1/\eta)}{\delta}\right)=O\!\left(
 \frac{\log(1/\eps)}{\gamma\Delta}
 +\frac{\log(1/\eta)}{\Delta}\right).
\]
While we focused here on upper bounding the number of applications of $U$ and $U^\dagger$, we can also upper bound the number of applications of $A$ and $A^\dagger$: the algorithm~$\cal A$ from \cref{ssec:HMdWalgo} uses $O(3^L)=O(1/\gamma)$ applications of $A$ and $A^\dagger$, and the algorithm of this subsection applies $\cal A$ at most $R=O(\log(1/\eps))$ times, giving the upper bound of $O(\log(1/\eps)/\gamma)$ that is claimed in \cref{thm:main}.

\section{An algorithm using transducers}\label{sec:transducerbasedproof}

An alternative optimal ground-state preparation algorithm can be derived using transducers. It {is arguably} a bit more involved than the one given in \cref{sec:HMWbasedproof}, but it follows a very natural strategy, and we expect the techniques we develop to be useful in other settings as well.
We first describe an algorithm that uses standard techniques, but does incur a suboptimal logarithmic factor.
We have access to $U$ and $U^\dagger$, and to a state preparation unitary $A$ (and $A^\dagger$) with $A\ket{0}=\ket{\psi}$, as described in \Cref{sec:setup}. Assume without loss of generality that $\abs{E_0} \leq \delta$, which we can always achieve if we shift everything by $\tilde{E}_0$ such that $|\tilde {E}_0-E_0|\leq \delta$. Recall the setting we work in: we assume that the next eigenphase satisfies $E_1-\tilde{E}_0 \geq E_0+3\delta-\tilde{E}_0\geq 2\delta$, and we can prepare a state $\ket{\psi}$ that has overlap at least $\gamma$ with the ground state.
The algorithm developed in this section aims at small constant trace distance and small constant error probability; the symbols $\eps,\eta$ are used in this section with local meanings different from the introduction (\cref{def:groundstateprep}).

Now, consider a real-valued function $f_\delta$ on $[-\pi,\pi)$ which is such that:
\begin{enumerate}
    \item $\abs{f_\delta(t)}$ is very small for $t \geq 2\delta$ (ideally zero, but $O(\gamma)$ will be good enough),
    \item $f_\delta(t)$ is larger than some constant for $\abs{t} \leq \delta$,
    \item\label{it:fourier} $f_\delta$ is a trigonometric polynomial of degree $J$, i.e. its Fourier series can be written as
    \begin{align*}
        f_\delta(t) = \sum_{j=-J}^J b_j e^{ijt},
    \end{align*}
    with normalization $\sum_j \abs{b_j} = 1$.
\end{enumerate}
Such functions exist with $J = \Theta(\delta^{-1}\log(\gamma^{-1}))$, see e.g. \cite{linlin&tong:groundstateprep}.
By the normalization condition in \cref{it:fourier}, we can implement a linear combination of unitaries (LCU), using an ancilla register $\C^{2J+1}$
\begin{align}\label{eq:lcu-sketch}
    \ket{\psi}\ket{0} \mapsto \Bigl( \sum_j b_j U^j \ket{\psi} \Bigr)\ket{0} + \ket{\psi^\perp}
\end{align}
where $\ket{\psi^\perp}$ is orthogonal to $\ket{0}$ in the second register. We call the subspace where the second register is in the state $\ket{0}$ the `good' subspace.
If we apply this to an eigenstate $\ket{\psi_k}$ with eigenphase $E_k$, this becomes
\begin{align*}
    \ket{\psi_k}\ket{0} \mapsto f_\delta(E_k)\ket{\psi_k}\ket{0} + \ket{\psi^\perp},
\end{align*}
so by the properties of $f_\delta$, and our assumptions on the spectrum, on the good subspace, this implements an approximate projection onto the ground space.
When applied to $\ket{\psi}$, the resulting amplitude is at least a constant times the overlap with the ground state, so $O(\gamma^{-1})$ rounds of amplitude amplification suffice to prepare the ground state with constant error, and each round uses $U$ and $U^\dagger$ each $J = \Theta(\delta^{-1}\log(\gamma^{-1}))$ times, and the state preparation $A, A^\dagger$ $O(1)$ times.

The reason for the suboptimality of this approach is that we need to truncate the degree $J$. However, when we implement $\sum_j b_j U^j$, if for large $j$ the $b_j$ are small, this intuitively means we only use $U^j$ with a low weight, for which we might not like to pay full price. Transducers make this intuition precise, and allow us to implement \cref{eq:lcu-sketch} with complexity $\sum_j \abs{j} \abs{b_j}$, so it is even acceptable to have infinite Fourier expansions, as long as the coefficients have sufficient decay. This way of counting complexity costs is different from the usual one, but can be realized in quantum algorithms, by using transducers, as we explain below.

\subsection{Transducers}\label{sec:transducers}

In this section, preliminary to the main content of \Cref{sec:transducerbasedproof}, we review \emph{transducers}, a model of quantum computation introduced in \cite{belovs2024taming}. Roughly speaking, a transducer represents a bounded-error quantum algorithm, in a way that is very convenient for composition. In particular, a transducer can often represent an algorithm in an exact way, even when it is only possible to implement this algorithm with error. See, for example~\cite[Section~13]{belovs2024taming} or \cite{belovs2024purifier}. This allows bounded-error algorithms to be composed without log factors, overcoming the issue described in the previous section. 

\paragraph{Transducers definition.} Let $S$ be a unitary acting on a direct sum of spaces, ${\cal H}\oplus {\cal L}$. Then by \cite[Theorem~3.1]{belovs2024taming}, there exists a unitary $U_S$ on ${\cal H}$ such that for any $\ket{\xi}\in {\cal H}$, there exists a vector $\ket{v}\in {\cal L}$ such that
$$S(\ket{\xi}+\ket{v})=U_S\ket{\xi}+\ket{v}.$$
Moreover, for any $\ket{\xi}\in {\cal H}$ and $\ket{v}\in{\cal L}$, if $S(\ket{\xi}+\ket{v})=\ket{\tau}+\ket{v}$, for $\ket{\tau}\in{\cal H}$, then $\ket{\tau}=U_S\ket{\xi}$. We call such a vector $\ket{v}$ a \emph{catalyst} for $\ket{\xi}$. A catalyst need not be unique, and may have any norm, but there is a unique smallest catalyst. 

The unitary $S$, along with the decomposition of its space into ${\cal H}$, which we call the \emph{public space}, and ${\cal L}$, which we call the \emph{private space}, is called a \emph{transducer}, and it can be seen to ``represent'' the computation $U_S$, which we call its \emph{transduction action}, and write $S:\ket\xi\rightsquigarrow\ket\tau$, or $\ket\xi\overset{S}{\rightsquigarrow}\ket\tau$, and say that $S$ \emph{transduces} $\ket\xi$ into $\ket\tau$, for $\ket{\tau}=U_S\ket{\xi}$. Indeed, we think of $U_S$ as some action we would like to perform, and $S$ is often an action that is much easier to implement. It is not obvious, but it turns out that by making sufficiently many calls to $S$, we can approximate the action $U_S$.

\begin{theorem}[Informal]\label{thm:informal}
    For a transducer $S$ and error parameter $\eps$, there is a quantum algorithm that implements $U_S$ with error $\eps$ using $K$ applications of $S$ and $O(K)$ other one- and two-qubit gates, as long as the input $\ket{\xi}$ has a catalyst $\ket{v}$ such that $K\geq 1+\frac{4}{\eps^2}\norm{\ket{v}}^2$.
\end{theorem}
Note that the catalyst $\ket{v}$ very much depends on the input $\ket{\xi}$, including potentially its size (squared norm) which can be any non-negative real number.

\paragraph{Transduction complexity.} As we see from the above theorem, the catalyst size $\norm{\ket{v}}^2$ governs the number of calls to $S$ needed to implement $U_S$, and we therefore define the \emph{transduction complexity} of $S$ with respect to a particular state, $W(S,\ket{\xi})$, as the minimum $\norm{\ket{v}}^2$ such that $\ket{v}$ is a catalyst for $\ket{\xi}$, i.e. $S(\ket{\xi}+\ket{v})=U_S\ket{\xi}+\ket{v}$. We let $W(S)$ be the supremum of $W(S,\ket{\xi})$ over all unit vectors $\ket{\xi}$. Note that this implies that for any unit vector $\ket{\xi}\in {\cal H}$, there exists a catalyst $\ket{v}\in {\cal L}$ such that
$$S(\ket{\xi}+\ket{v})=U_S\ket{\xi}+\ket{v}\quad\mbox{and}\quad \norm{\ket{v}}^2\leq W(S).$$
Moreover, for any vector $\ket{\xi}\in {\cal H}$, not necessarily normalized, there exists a vector $\ket{v}\in {\cal L}$ such that
$$S(\ket{\xi}+\ket{v})=U_S\ket{\xi}+\ket{v}\quad\mbox{and}\quad \norm{\ket{v}}^2\leq \norm{\ket{\xi}}^2 W(S).$$
{Since any catalyst upper bounds $W(S,\ket\xi)$, we will typically exhibit a particular catalyst rather than the smallest one.}

Looking at \Cref{thm:informal}, we can see that the two measures of complexity that matter for a transducer $S$ -- i.e. that determine the complexity of approximating $U_S$ -- are $W(S)$ (or $W(S,\ket{\xi})$ if we want to restrict our attention to particular states) and the complexity of implementing the unitary $S$. This latter measure, called the \emph{iteration time} of $S$, is a model-dependent quantity, whose precise value depends on which unitaries are allowed as elementary operations. This might be some specific universal gate set, but it might also include oracle calls to a potentially more complicated unitary that we treat as a black box -- or multiple such oracles, which can be handled by the single oracle case by letting $O$ apply one of several oracles $O_1,\dots,O_c$ controlled on the value in some auxiliary register. 

\paragraph{Oracles and canonical form.} If the implementation of $S$ makes black box calls to some unitary oracle $O$, we can write it as $S(O)$. This is often desirable, because transducers are particularly handy for algorithmic composition, and we may wish to later replace $O$ with another transducer to get a transducer for the composed behaviour. We won't go into this in more detail here, but we will see an example in \Cref{sec:aa-transducer-second}.

Any implementation of $S(O)$ that depends non-trivially on $O$ must make at least one call to $O$, so if we apply \Cref{thm:informal} with $K=\Theta(W(S))$, we would make at least $\Theta(W(S))$ calls to $O$, which may not be desirable. It turns out we can do better than this. We can define a notion of query complexity in which the query complexity of $S(O)$ to $O$ can be potentially less than 1. In order to define \emph{Las Vegas query complexity}, which was first defined in~\cite{belovs2023LasVegas}, we need $S(O)$ to be in \emph{canonical form}, meaning:
\begin{enumerate}
    \item $S(O)=S^\circ O$, for some oracle-independent unitary $S^\circ$ called the \emph{work unitary};
    \item the private space can be decomposed ${\cal L}={\cal L}^\bullet\oplus {\cal L}^\circ$, and $O$ acts trivially on ${\cal H}$ and ${\cal L}^\circ$. We call ${\cal L}^\bullet$ the \emph{query part} of the space.
\end{enumerate}
It is possible to modify any transducer to be in canonical form~\cite[Proposition~10.4]{belovs2024taming} (we give an example in \Cref{sec:aa-transducer-second}). For any canonical form transducer $S$ on ${\cal H}\oplus {\cal L}$, and input $\ket{\xi}\in {\cal H}$, we define the Las Vegas query complexity as:
$$L(S,O,\ket{\xi})=\norm{\Pi_{{\cal L}^\bullet}\ket{v}}^2,$$
where $\ket{v}$ is the chosen catalyst for $\ket{\xi}$. Then we let $L(S,O)$ be the supremum of the Las Vegas query complexity over all unit vectors $\ket{\xi}$.
As in \cite[Section~5.1]{belovs2024taming}, we will in fact usually fix a particular catalyst $\ket v$ for $\ket\xi$ and write $W(S(O),\ket\xi)=\norm{\ket v}^2$ and $L(S,O,\ket\xi)=\norm{\Pi_{{\cal L}^\bullet}\ket v}^2$ for that catalyst. \Cref{thm:transducer-to-alg} below holds whenever there is a single catalyst satisfying both bounds.

To motivate this notion, consider a classical Las Vegas algorithm that makes a query to an oracle $O$ with probability $p$. Then the expected number of calls to $O$ is $p$. Similarly, if we make a controlled call to $O$ with amplitude only $\sqrt{p}\leq 1$, then the Las Vegas query complexity is just $p$. What is perhaps surprising is that this optimistic way of counting queries can actually be realized in actual algorithms. In transducer composition, the transduction complexity of an inner transducer is not multiplied by the transduction complexity of the outer, $\norm{\ket{v}}^2$, but by the Las Vegas complexity, $\norm{\Pi_{{\cal L}^\bullet}\ket{v}}^2$~\cite[Proposition~9.11]{belovs2024taming}
(something similar happens even for non-canonical form transducers, though it is more complicated to describe). We also have the following potentially improved version of \Cref{thm:informal}, which follows from~\cite[Theorem~3.3]{belovs2024taming}:
\begin{theorem}\label{thm:transducer-to-alg}
    Let $S(O)=S^\circ O$ be a canonical form transducer, $\eps\in (0,1)$, and $W$ and $L$ non-negative real numbers. There exists an algorithm that uses $O(L/\eps^2)$ controlled calls to $O$, $O(1+W/\eps^2)$ controlled calls to $S^\circ$, and $O(1+W/\eps^2)$ additional one- and two-qubit gates; and on input $\ket{\xi}$, outputs a state $\ket{\tilde\tau}$ such that if $W(S(O),\ket{\xi})\leq W$ and $L(S,O,\ket{\xi})\leq L$, $\norm{U_S\ket{\xi}-\ket{\tilde\tau}}\leq \eps$.
\end{theorem}

\paragraph{Composition.} Composing transducers is often much more efficient than composing quantum algorithms. One reason is the previously discussed Las Vegas query complexity -- if you only apply a subroutine to a small part of the state, the cost you pay for it is scaled down accordingly. One consequence of this is that if you are running a subroutine on a superposition of different inputs, and in many branches of the superposition you have terminated early and are doing nothing, then you do not pay for this -- you only pay the average cost.
Another reason is that transducers can often represent a computation exactly, even when we only know how to implement it with bounded error (this is not a contradiction, since the algorithms from \Cref{thm:informal,thm:transducer-to-alg} only approximate $U_S$). There are various general theorems about composing transducers~\cite[Section~9]{belovs2024taming}, 
but in this paper, we will show all compositions explicitly, which is possible since our compositions are not all that complicated, and hopefully makes the strange subject of transducer composition somewhat less mysterious.

\subsection{Filter function}\label{sec:filter}
The high-level idea is that we will apply a function $f$ to the unitary $U$ such that $f(U)$ only has support near the ground state energy. We now construct an explicit choice of function that has the right properties. This choice is not unique; it mostly matters that it has support contained in $[-2\delta,2\delta]$, and is sufficiently smooth.

We start from a function $g_\delta$, for $0 < \delta < \pi/2$, defined on $[-\pi,\pi)$ by
\begin{align}\label{eq:window1}
    g_\delta(t) \coloneqq \begin{cases}
        2\sqrt{\omega}\cos(\omega t) & \abs{t} \leq \delta,\\
        0 & \mbox{otherwise},
    \end{cases}
    \qquad\mbox{where }\omega = \frac{\pi}{2\delta}.
\end{align}
This function has been used before as a tapering function for phase estimation \cite{rendon2022effects}.
We then let
\begin{align}
    f_\delta(t) \coloneqq \frac{1}{2\pi}\int_{-\pi}^\pi g_\delta(t - s) g_\delta(s) \, ds.
\end{align}
From the definition, an easy calculation gives
\begin{equation}\label{eq:f-delta}
    f_\delta(t) = \begin{cases}
        (1 - \frac{\abs{t}}{2\delta})\cos(\omega t) + \frac{1}{\pi} \sin(\omega\abs{t}) & \abs{t} \leq 2\delta,\\
        0 & \mbox{otherwise}.
        \end{cases}
\end{equation}
This function is such that $f_\delta(t) \geq \frac{1}{\pi}$ for $\abs{t} \leq \delta$, and it is zero for $\abs{t} \geq 2\delta$.
We use the Fourier transform $\mathcal F: L^2([-\pi,\pi)) \to \ell^2(\Z)$, $f \mapsto \hat f$ with convention
\begin{equation}\label{eq:f-transform}
    \hat f[j] = \frac{1}{2\pi} \int_{-\pi}^\pi f(t) e^{-ijt} \, dt, \quad \text{ so } \quad f(t) = \sum_{j \in \Z} \hat f[j] e^{ijt} \, .
\end{equation}

We now record some basic properties of the Fourier transform of $g_\delta$.

\begin{lemma}\label{lem:fourier-g}
    Let $c_j = \hat g_\delta[j]$. Then for $j\neq\pm\omega$,
        $\displaystyle c_j=\frac{2\omega^{3/2}\cos(j\delta)}{\pi(\omega^2-j^2)}$, extended by continuity at $j=\pm\omega$.
    These Fourier coefficients satisfy
    \begin{enumerate}
        \item $\sum_j \abs{c_j}^2 = 1$,
        \item $\sum_j \abs{j} \abs{c_j}^2 \leq \omega$.
    \end{enumerate}
\end{lemma}

\begin{proof}
    Since $g_\delta$ is real and even, $c_j = c_{-j}$ is real. This implies
    \begin{align*}
        c_j = \frac{1}{2\pi}\int_{-\pi}^\pi g_\delta(t)e^{-ijt}\,dt=\frac{\sqrt{\omega}}{\pi} \int_{-\delta}^{\delta}\cos(\omega t)\cos(jt)\,dt.
    \end{align*}
    The product-to-sum identity allows us to evaluate the integral as
    \begin{align*}
        \int_{-\delta}^{\delta}\cos(\omega t)\cos(jt)\,dt
        &= \frac{\sin((\omega-j)\delta)}{\omega-j}+\frac{\sin((\omega+j)\delta)}{\omega+j} \\
    &=\cos(j\delta)\left(\frac{1}{\omega-j}+\frac{1}{\omega+j}\right)
    =\frac{2\omega\cos(j\delta)}{\omega^2-j^2},
    \end{align*}
    where the middle step uses $\omega\delta=\pi/2$, so that $\sin((\omega\mp j)\delta)=\sin(\pi/2\mp j\delta)=\cos(j\delta)$. This yields the claimed expression for $c_j$.
    The Parseval identity states that
    \begin{align*}
        \sum_j \abs{c_j}^2 = \norm{\hat g_\delta}^2 = \frac{1}{2\pi} \norm{g_\delta}^2,
    \end{align*}
    and it is easy to see that
    \begin{align*}
        \norm{g_\delta}^2 = \int_{-\delta}^\delta 4 \omega \cos^2(\omega t) \, dt = 4 \omega \delta = 2\pi\, .
    \end{align*}
    Next, we note that $g_\delta$ is continuous and piece-wise differentiable, and $g_\delta'$ has Fourier coefficients $i j c_j$. Another application of Parseval yields
    \begin{align*}
        \sum_j \abs{j}^2 \abs{c_j}^2 = \frac{1}{2\pi} \norm{g_\delta'}^2 = \frac{1}{2\pi} \int_{-\delta}^\delta 4 \omega^3 \sin^2(\omega t) \, dt = \frac{1}{2\pi} 4 \omega^3 \delta = \omega^2.
    \end{align*}
    Finally, by Cauchy-Schwarz,
    \begin{align*}
        \sum_j \abs{j} \abs{c_j}^2 \leq \left(\sum_j \abs{c_j}^2\right)^{\frac12} \left(\sum_j \abs{j}^2 \abs{c_j}^2\right)^{\frac12} \leq \omega \, .
    \end{align*}
\end{proof}

We let $p_j = \hat f_\delta[j]$ denote the Fourier coefficients of $f_\delta$.
Since $f_\delta = \frac{1}{2\pi}g_\delta \ast g_\delta$ is a convolution, we have $p_j = c_j^2$. Thus, the following is a direct consequence of \cref{lem:fourier-g}.

\begin{lemma}\label{lem:fourier-f}
    Let $p_j = \hat f_\delta[j]$. Then $p_j = c_j^2\geq0$, and we have
    \begin{align*}
        \sum_j p_j = 1, \qquad \sum_j \abs{j} p_j \leq \omega=\frac{\pi}{2\delta}.
    \end{align*}
    In particular, $(p_j)_{j\in\Z}$ is a probability distribution.
\end{lemma}

\subsection{A transducer for a linear combination of unitaries}\label{sec:lcu-transducer}

Let $p = (p_j)_{j \in \Z}$ be a probability distribution, and $U$ a unitary on a Hilbert space $\C^N$. We now construct a transducer for the LCU transformation on $\C^N \ot \ell^2(\Z)$
\begin{align}\label{eq:lcu-target}
    \ket{\psi}\ket{0} \mapsto \left(\sum_{j \in \Z} p_j U^j \right) \ket{\psi} \ket{0} + \ket{\psi^\perp}
\end{align}
for arbitrary $\ket{\psi} \in \C^N$, where $\ket{\psi^\perp}$ is orthogonal to $\ket{0}$ on the second register and need not be normalized. Specifically, in the remainder of this section, we prove the following:
\begin{theorem}\label{thm:LCU-transducer}
Let $P$ be a unitary on $\ell^2(\Z)$ that satisfies $P\ket{0} = \sum_j \sqrt{p_j} \ket{j}$, often called `Prepare'.
    There is a transducer $S_F$ acting on $\mathbb{C}^N\otimes \ell^2(\Z)\otimes \ell^2(\N_0)$ with public space ${\cal H}_F=\C^N\otimes \ell^2(\Z)\otimes\mathrm{span}\{\ket{0}\}$ such that:
    \begin{enumerate}
        \item For any unit vector $\ket{\psi}\in \C^N$, $\ket{\psi}\ket{0}\ket{0}\overset{S_F}{\rightsquigarrow }\left((\sum_{j\in\Z}p_j U^j)\ket{\psi}\ket{0}+\ket{\psi^\bot}\right)\ket{0}$ for some $\ket{\psi^\bot}$ orthogonal to $\C^N\otimes\mathrm{span}\{\ket{0}\}$, and $W(S_F,\ket{\psi}\ket{0}\ket{0})\leq \sum_j p_j \abs{j}$.
        \item $S_F$ can be implemented using one call to each of $P$ and $P^\dagger$, one controlled call to each of $U$ and $U^\dagger$ and one controlled incrementation of the last register modulo the value $\abs{j}$ in the second register.
    \end{enumerate}
\end{theorem}

The transducer is a modification of a standard construction, which we briefly recall. Assume that the sum only runs to $\abs{j} \leq J$, and use $\C^{2J+1}$ with basis $\ket{j}$ for $j = -J, \dots, J$ instead of $\ell^2(\Z)$. Consider the unitary
\begin{align*}
    \mathsf{Sel} = \sum_{j = -J}^J U^j \ot \ketbra{j}{j},
\end{align*}
called `Select'. Then $(I \ot P^\dagger)\mathsf{Sel}(I \ot P)$ implements the transformation in \cref{eq:lcu-target}.
Note that it uses $U$ and $U^\dagger$ $J$ times each. However, higher powers $j$ may occur with low weight $p_j$. It is straightforward to make a transducer that uses a single controlled query to $U$ and to $U^\dagger$, and has transduction complexity bounded by $\sum_j p_j \abs{j}$.

We build this transducer on the space $\C^N \ot \ell^2(\Z) \ot \ell^2(\N_0)$, where the public and private spaces are
\begin{align*}
    \H = \C^N \ot \ell^2(\Z) \ot \mathrm{span}\{\ket{0}\} \quad \text{ and } \quad \L = \C^N \ot \ell^2(\Z) \ot \mathrm{span}\{\ket{k}, \, k \geq 1\}\, .
\end{align*}
Let $Q$ be a unitary on $\ell^2(\Z)$ that maps $Q\ket{0} = \sum_j \sqrt{p_j}\ket{j}$, and let $P$ be the unitary on $\H \oplus \L$ that applies $Q$ controlled on the last register being in the state $\ket{0}$ (i.e. it is applied on the public space). Let $\mathrm{sgn}(j)$ be the sign of $j$, and equal to 0 if $j=0$.
Let
\begin{align*}
    \mathsf{Sel} = \sum_{j \in \Z} \Bigl(\sum_{k = 0}^{\abs{j}-1} U^{\mathrm{sgn}(j)} \ot \ketbra{j}{j} \ot \ketbra{k+1 \text{ mod } \abs{j}}{k}\;+\;\sum_{k\geq\abs{j}} I\ot\ketbra jj\ot\ketbra kk\Bigr),
\end{align*}
where the $j=0$ term should be interpreted as $I \ot \ketbra{0}{0} \ot I_{\ell^2(\N_0)}$.
Similarly, in sums of the form $\sum_{k=0}^{\abs j-1}$ below, the term $j=0$ is interpreted as the single term $k=0$.
We define
\begin{align}\label{eq:transducer-lcu}
    S = P^\dagger \, \mathsf{Sel} P
\end{align}
and will show that this is a transducer for the transformation in \cref{eq:lcu-target}. This proves \cref{thm:LCU-transducer} with $S_F\defeq S$.

\begin{lemma}\label{lem:LCU-transducer}
    The unitary $S$ in \cref{eq:transducer-lcu} transduces
    \begin{align}\label{eq:transducer-lcu-2}
        \ket{\psi}\ket{0}\ket{0} \rightsquigarrow \left(\sum_{j \in \Z} p_j U^j \right) \ket{\psi} \ket{0}\ket{0} + \ket{\psi^\perp}
    \end{align}
    where $\ket{\psi^\perp}$ is orthogonal to $\ket{0}$ on the second register and need not be normalized. It has transduction complexity
    \begin{align*}
        W(S, \ket{\psi}\ket{0}\ket{0}) \leq \sum_{j \in \Z} p_j \abs{j} \,.
    \end{align*}
    Additionally, $S^\dagger$ transduces
    \begin{align}\label{eq:transducer-lcu-2-dagger}
        \ket{\psi}\ket{0}\ket{0} \rightsquigarrow \left(\sum_{j \in \Z} p_j U^{-j} \right) \ket{\psi} \ket{0}\ket{0} + \ket{\psi'^\perp}
    \end{align}
    where $\ket{\psi'^\perp}$ is orthogonal to $\ket{0}$ on the second register, and has the same transduction complexity.
\end{lemma}

\begin{proof}
    The catalyst is the following history-like state
    \begin{align*}
        \ket{v} = \sum_{j \in \Z} \sqrt{p_j} \sum_{k=1}^{\abs{j}-1} U^{\mathrm{sgn}(j)k} \ket{\psi} \ket{j} \ket{k}.
    \end{align*}
    Note that $\ket{v} \in \L$, and $\norm{\ket v}^2=\sum_{j\neq0}p_j(\abs j-1)\norm{\ket\psi}^2\leq\sum_j p_j\abs j$ for a unit vector $\ket\psi$.
    Abbreviate the source and target states in \cref{eq:transducer-lcu-2} as $\ket{s}$ and $\ket{t}$, then we need to show that
    \begin{align*}
        S(\ket{s} + \ket{v}) = \ket{t} + \ket{v}.
    \end{align*}
    First of all, recall that $P$ only acts on the public space, so $P(\ket{s} + \ket{v}) = \sum_j \sqrt{p_j}\ket{\psi}\ket{j}\ket{0} + \ket{v}$. We then see that
    \begin{align*}
        \ket{s} + \ket{v} &\mapsto_{P} \sum_j \sqrt{p_j}\ket{\psi}\ket{j}\ket{0} + \sum_{j \in \Z} \sqrt{p_j} \sum_{k=1}^{\abs{j}-1} U^{\mathrm{sgn}(j)k} \ket{\psi} \ket{j} \ket{k} \\
        &= \sum_{j \in \Z} \sqrt{p_j} \sum_{k=0}^{\abs{j}-1} U^{\mathrm{sgn}(j)k} \ket{\psi} \ket{j} \ket{k}\\
        &\mapsto_{\mathsf{Sel}} \sum_{j \in \Z} \sqrt{p_j} \sum_{k=0}^{\abs{j}-1} U^{\mathrm{sgn}(j)}U^{\mathrm{sgn}(j)k} \ket{\psi} \ket{j} \ket{k+1 \text{ mod } \abs{j}} \\
        &= \sum_{j \in \Z} \sqrt{p_j} U^{j} \ket{\psi} \ket{j} \ket{0} + \underbrace{\sum_{j \in \Z} \sqrt{p_j} \sum_{k=1}^{\abs{j}-1} U^{\mathrm{sgn}(j)k} \ket{\psi} \ket{j} \ket{k}}_{\ket{v}} \\
        &\mapsto_{P^\dagger} \left(\sum_{j \in \Z} p_j U^j \right) \ket{\psi} \ket{0}\ket{0} + \ket{\psi^\perp} + \ket{v} = \ket{t} + \ket{v}
    \end{align*}
    again using in the last step that $P^\dagger$ only acts nontrivially on $\H$, and that $\bra0Q^\dagger\ket j=\sqrt{p_j}$.
    The claim about $S^\dagger$ follows from a similar calculation, using $\ket{v'}=\sum_{j\in\Z}\sqrt{p_j}\sum_{k=1}^{\abs{j}-1}U^{-\mathrm{sgn}(j)(\abs{j}-k)}\ket{\psi}\ket{j}\ket{k}$ as a catalyst.
\end{proof}

Note that if $p$ only has support on $\{-J, -J+1, \dots, J-1, J\}$, we can truncate $\ell^2(\Z)$ and $\ell^2(\N_0)$ to $\C^{2J+1}$ and $\C^{J}$ without error. However, the transducer of \cref{thm:LCU-transducer} is sensible also for infinite expansions, as long as the first moment $\sum_j p_j \abs{j}$ is finite, which is not the case for the usual LCU construction.

While not needed now, it is easy to modify the above construction to implement $\sum_j q_j U^j$ where $q_j \in \C$ and $\sum_j \abs{q_j}=1$, simply by adding the phases in $\mathsf{Sel}$.

The idea is to apply the LCU transducer with $U = e^{iH}$, and the $p_j$ from \cref{sec:filter}.
Note that in that case,
\begin{align*}
    \sum_{j \in \Z} p_j U^j = \sum_{j \in \Z} p_j e^{i j H} = f_\delta(H)
\end{align*}
so we are applying a filter on the spectrum of the Hamiltonian.
\cref{thm:LCU-transducer} and \cref{lem:fourier-f} directly yield:

\begin{corollary}\label{cor:infinite-filter-transducer}
        There is a transducer $S_F$ acting on $\mathbb{C}^N\otimes \ell^2(\Z)\otimes \ell^2(\N_0)$ with public space ${\cal H}_F=\C^N\otimes \ell^2(\Z)\otimes\mathrm{span}\{\ket{0}\}$ such that:
    \begin{enumerate}
        \item For any eigenstate $\ket{\psi_k}\in \C^N$ of $H$ with energy $E_k$, $\ket{\psi_k}\ket{0}\ket{0}\overset{S_F}{\rightsquigarrow }\left(f_\delta(E_k) \ket{\psi_k}\ket{0}+\ket{\psi_k^\bot}\right)\ket{0}$ for some $\ket{\psi_k^\bot}$ orthogonal to $\C^N\otimes\mathrm{span}\{\ket{0}\}$, and $W(S_F,\ket{\psi_k}\ket{0}\ket{0}) \leq \frac{\pi}{2\delta}$.
        \item $S_F$ can be implemented using one controlled call to each of $U = e^{iH}$ and $U^\dagger=e^{-iH}$.
    \end{enumerate}
\end{corollary}

If we only care about the query complexity of the ground-state preparation algorithm (in terms of queries to $U$, $A$ and their inverses), we can work directly with an infinite expansion, and \Cref{cor:infinite-filter-transducer} suffices.
However, if we care about efficient implementations, using small finite space overhead, we do need to truncate the Fourier series at some finite degree $J$. Generically, implementing $P$ could require $\Theta(J)$ gates, but in our case we can explicitly construct the amplitudes, since they are the Fourier coefficients of a simple function.

\begin{theorem}\label{thm:finite-lcu-transducer}
    Choose $\delta$, and integer $J$ such that $(2J+1)\delta/\pi$ is integer. Then for $J = \Theta(\delta^{-3/2} \eta^{-1})$ there exists a transducer $S_F$ acting on $\mathbb{C}^N\otimes \C^{2J+1} \otimes \C^J$ with public space ${\cal H}_F=\C^N\otimes \C^{2J+1} \otimes\mathrm{span}\{\ket{0}\}$ and a function $\tilde f_\delta$, satisfying $\norm{\tilde f_\delta(H) - f_\delta(H)}_{\infty} \leq \eta$ such that:
    \begin{enumerate}
        \item For any $\ket{\psi}\in \C^N$, $\ket{\psi}\ket{0}\ket{0}\overset{S_F}{\rightsquigarrow } \left(\tilde f_\delta(H)\ket{\psi}\ket{0}+\ket{\psi^\bot}\right)\ket{0}$ for some $\ket{\psi^\bot}$ orthogonal to $\C^N\otimes\mathrm{span}\{\ket{0}\}$, and $W(S_F,\ket{\psi}\ket{0}\ket{0}) = O(\delta^{-1})$ and $W(S_F^\dagger,\ket{\psi}\ket{0}\ket{0}) = O(\delta^{-1})$.
        \item $S_F$ can be implemented using one controlled call to each of $U = e^{iH}$ and $U^\dagger = e^{-iH}$ and $O((\log J)^2)$ other operations and $O(\log J)$ ancilla qubits.
    \end{enumerate}
\end{theorem}

\begin{proof}
    The construction of $S_F$ is the same as in \Cref{thm:LCU-transducer}, except that we work with truncated Fourier coefficients.
    We first observe that if we want to implement $\sum_j p_j U^j$, for $p_j \geq 0$, it is also acceptable to have $P\ket{0} = \sum_j e^{i\phi_j} \sqrt{p_j}\ket{j}$ for some arbitrary phases, as these will cancel with the application of $P^\dagger$. It therefore will suffice to prepare (an approximation of) the state $\sum_j c_j \ket{j}$, and we do so by using that the $c_j$ are the Fourier coefficients of a simple function.
    We use the following construction:
    \begin{enumerate}
        \item Identify $\C^{2J+1}$ with a uniform grid of  $2J+1$ points $t_a$ on the unit circle $[-\pi,\pi)$. Prepare a uniform superposition over the subset $A_\delta$ of $(2J+1)\delta/\pi$ points in the interval $[-\delta,\delta)$.
        \item Take one ancilla qubit, and apply $\mathsf{Had} \, V \, \mathsf{Had} $, where $\mathsf{Had}$ is a Hadamard on the ancilla, and $V$ the unitary that applies $e^{\pm i \omega t_a}$, with the ancilla controlling the sign. This prepares a state
        \begin{align*}
            \frac{1}{\sqrt{\abs{A_\delta}}}\ket{0}\Bigl(\sum_{a \in A_\delta} \cos(\omega t_a) \ket{a} \Bigr) + \frac{i}{\sqrt{\abs{A_\delta}}} \ket{1}\Bigl(\sum_{a \in A_\delta} \sin(\omega t_a) \ket{a} \Bigr).
        \end{align*}
        This state has squared amplitude exactly 1/2 on the state with $\ket{0}$ on the ancilla, so exact amplitude amplification yields the state
        \begin{align*}
            \ket{\phi} = \sqrt{\frac{2}{\abs{A_\delta}}} \sum_{a \in A_\delta} \cos(\omega t_a) \ket{a} = \sum_{a} \frac{g_\delta(t_a)}{\sqrt{2J+1}} \ket{a}
        \end{align*}
        \item Apply the quantum Fourier transform to get a state $\sum_{j=-J}^J \tilde c_j \ket{j}.$        
    \end{enumerate}
    First of all, this can be done in $O(\log(J)^2)$ steps and $O(\log J)$ ancilla qubits (the dominant cost being the quantum Fourier transform, for an exact implementation choose $2J + 1$ to be a power of two).
    If we let $\tilde f_\delta(t) = \sum_j \abs{\tilde c_j}^2 e^{ijt}$, then this construction gives an exact transducer for $\sum_{j=-J}^J \abs{\tilde c_j}^2 U^j = \tilde f_\delta(H)$. We now show this is a good approximation for large enough $J$.
    By basic Fourier analysis, since the coefficients of $\ket{\phi}$ consist precisely of the function $g_\delta$ sampled at a uniform grid, we have
    \begin{align*}
        \tilde c_j = \sum_{r \in \Z} c_{j + r(2J+1)}
    \end{align*}
    for $j = -J, -J+1, \dots, J-1,J$.
    Assume $J > 2\omega$.
    By \Cref{lem:fourier-g}, for $|j| > 2\omega$, we have $c_j = O(\delta^{-3/2} j^{-2})$, and in particular (letting $\tilde c_j = 0$ for $\abs{j} > J$) we have error $\abs{\tilde c_j - c_j} = O(\delta^{-3/2}J^{-2})$ for $\abs{j} \leq J$ and $\abs{\tilde c_j - c_j} = O(\delta^{-3/2}j^{-2})$ for $\abs{j} > J$.
    We conclude that
    \begin{align*}
        \abs{\tilde f_\delta(t) - f_\delta(t)} &\leq \sum_{j \in \Z} \abs{ \tilde{c_j}^2 - c_j^2} \leq 2 \sum_{j \in \Z} \abs{\tilde c_j - c_j} \\
        &\leq 2 \sum_{\abs{j} \leq J} \abs{\tilde c_j - c_j} + 2 \sum_{\abs{j} > J} \abs{\tilde c_j - c_j}\\
        &= O(\delta^{-3/2}J^{-2} \cdot J) + O\Bigl(\delta^{-3/2}\sum_{j=J+1}^{\infty} j^{-2}\Bigr) = O(\delta^{-3/2}J^{-1}).
    \end{align*}
    Thus, for $J = \Theta(\delta^{-3/2} \eta^{-1})$, we can achieve $\norm{\tilde f_\delta(H) - f_\delta(H)}_{\infty} \leq \eta$.
    It remains to bound $W(S_F,\ket{\psi}\ket{0}\ket{0})$.
    This follows from
    \begin{align*}
        W(S_F,\ket{\psi}\ket{0}\ket{0}) &\leq \sum_{\abs j\leq J}\abs j\abs{\tilde c_j}^2 \leq 2\sum_{\abs j\leq J}\abs j\abs{c_j}^2 + 2\sum_{\abs j\leq J}\abs j\abs{\tilde c_j-c_j} ^2 \\
        &\leq 2\omega + O(\delta^{-3}J^{-2})=O(\delta^{-1}),
    \end{align*}
using
$\abs{a}^2\leq2\abs{b}^2+2\abs{a-b}^2$, \Cref{lem:fourier-g}, the bound
on $\abs{\tilde c_j-c_j}$ above, and $J\geq\delta^{-1}$.
    This establishes the bound on $W(S_F,\ket{\psi}\ket{0}\ket{0})$; the same bound holds for $W(S_F^\dagger,\ket{\psi}\ket{0}\ket{0})$, as in \Cref{lem:LCU-transducer}.
\end{proof}

As a side comment, we note that we could have taken a smoother function than $f_\delta$, which would have led to faster decay of the Fourier coefficients, and could reduce the required $J$. However, this does not affect the asymptotic complexity, as $J$ only appears as a logarithmic factor in the number of ancilla qubits and additional operations.

\subsection{Transducers for amplitude amplification}\label{sec:aa-transducers}

In this section, we give a transducer for amplitude amplification, which is the following problem.

\begin{definition}[Amplitude amplification]
    Given an oracle $A'$ on ${\cal V}$ such that $A'\ket{0}=\ket{\pi}$, and an oracle ${R}=2\Pi_M-I$ for some orthogonal projector $\Pi_M$ on ${\cal V}$, output $\ket{\pi_M}:=\frac{1}{\norm{\Pi_M\ket{\pi}}}\Pi_M\ket{\pi}$.
\end{definition}
Throughout this section, we assume $\mu\in (0,1/2]$, where
$$\mu\defeq\norm{\Pi_M\ket\pi}^2,\quad\mbox{so}\quad\ket{\pi_M}=\frac{\Pi_M\ket\pi}{\sqrt\mu},\quad\mbox{and define}\quad \ket{\pi_{\bar M}}\defeq\frac{(I-\Pi_M)\ket\pi}{\sqrt{1-\mu}}.$$
We note that $\mu\leq \frac{1}{2}$ is not a real restriction. The $\mu>1/2$ case of amplitude amplification is trivial, and we can always ensure $\mu\leq \frac{1}{2}$ by taking any instance, adding an ancilla $\ket{+}$ and replacing $\Pi_M$ by $\Pi_M\otimes\ket{1}\bra{1}$, which halves $\mu$. 
We will not make this explicit.

We first give a transducer for amplitude amplification that uses $O(1)$ controlled calls to $R$, $A'$, and its inverse, and has transduction complexity $O(1/\sqrt{\mu})$. We prove the following theorem in \Cref{sec:aa-transducer-first}.
\begin{theorem}\label{thm:aa-transducer}
There is a transducer $S_{\rm AA}$ on ${\cal F}\ot{\cal C}\ot{\cal V}$, where ${\cal F}$ is a flag qubit and ${\cal C}=\mathrm{span}\{\ket\ell:\ell\in\N_0\}$, such that:
    \begin{enumerate}
        \item its public space is $(\mathrm{span}\{\ket{0}_{\cal F}\ket{0}_{\cal C}\}\ot{\cal V})\oplus(\mathrm{span}\{\ket1_{\cal F}\}\ot{\cal C}\ot{\cal V})$;
        \item {whenever $\mu>0$, }it has transduction action $S_{\rm AA}:\ket{0}_{\cal F}\ket{0}_{\cal C}\ket{\pi}\rightsquigarrow\ket{1}_{\cal F}\ket{\lambda}_{\cal C}\ket{\pi_M}$ for a unit vector $\ket\lambda\in{\cal C}$ that depends only on $\mu$;
        \item {$W(S_{\rm AA},\ket0\ket0\ket\pi)=O(1/\sqrt\mu)$;}
        \item one application of $S_{\rm AA}$ uses one controlled call to each of $A'$ and ${A'}^\dagger$, and $O(1)$ controlled calls to $R$.
    \end{enumerate}
\end{theorem}
Our transducer differs from the one in \cite[Theorem~6.1]{apers2026elfs} in that it uses fewer qubits. The transducer of \cite{apers2026elfs} is based on a coherent version of standard Las Vegas amplitude amplification, which runs standard amplitude amplification repeatedly (logarithmically many times in expectation), with geometrically decreasing guesses of $\mu$. Making this process coherent requires a distinct register for every repetition, resulting in a multiplicative log-factor overhead in the number of qubits. Our transducer's only qubit overhead is an additive overhead from the counter register ${\cal C}$. Admittedly, this register is infinite, but we can truncate it to get the following construction, for which there is also an efficient implementation:
\begin{theorem}\label{thm:aa-truncated-intro}
    For any $\eta\in (0,1)$, there exists a positive integer $D=O(1/(\eta\sqrt{\mu}))$ and a transducer $S_{\rm AA}^{[D]}$ on ${\cal F}\otimes {\cal C}^{[D]}\otimes {\cal V}$, where ${\cal F}$ is a flag qubit and ${\cal C}^{[D]}=\mathrm{span}\{\ket{0},\dots,\ket{D-1}\}$, such that:
    \begin{enumerate}
        \item its public space is $(\mathrm{span}\{\ket{0}\ket{0}\}\otimes {\cal V})\oplus (\mathrm{span}\{\ket{1}\}\otimes {\cal C}^{[D]}\otimes{\cal V})$;
        \item whenever $\mu>0$, it has transduction action $S_{\rm AA}^{[D]}:\ket{0}\ket{0}\ket{\pi}\rightsquigarrow\ket{\tau}_{{\cal FC}^{[D]}{\cal V}}$ for $\ket{\tau}$ such that $\norm{\ket{\tau}-\ket{1}\ket{\lambda'}\ket{\pi_M}}\leq \eta$ for some unit vector $\ket{\lambda'}\in \mathbb{C}^{[D]}$ that depends only on $\mu$ and $D$;
        \item $W(S_{\rm AA}^{[D]},\ket{0}\ket{0}\ket{\pi})=O(1/\sqrt{\mu})$;
        \item One application of $S_{\rm AA}^{[D]}$ uses one controlled call to each of $A'$ and ${A'}^\dagger$, $O(1)$ controlled calls to $R$, and $O(\log \frac{1}{\eta\sqrt{\mu}}+\log\dim{\cal V})$ additional one- and two-qubit gates.
    \end{enumerate}
\end{theorem}
\noindent This is proven in \Cref{sec:aa-transducer-first}, where it is restated as \Cref{thm:aa-truncated}.

\paragraph{A canonical transducer for amplitude amplification.}
In our specific setting, we will have $A'=U_FA$, where $A$ generates a guiding state, and $U_F$ applies a filter to dampen all but the ground state, and we have a transducer $S_F$ for $U_F$.  We want to make the transducer canonical in the oracle $A$ so that we can apply \Cref{thm:transducer-to-alg} to get a quantum algorithm that makes fewer calls to $A$. We can compose transducers and make them canonical using \cite[Section~9]{belovs2024taming} and \cite[Section~10.4]{belovs2024taming} respectively, but we make the construction explicit in \Cref{sec:aa-transducer-second}, where we prove \Cref{thm:aa-canonical}, stated shortly.

It is simple to see that if $S_F$ is a transducer for $U_F$, then $S_F^\dagger$ is a transducer for $U_F^\dagger$: if $S_F(\ket\xi+\ket v)=U_F\ket\xi+\ket v$, then $S_F^\dagger(U_F\ket\xi+\ket v)=\ket\xi+\ket v$. Our analysis will only depend on catalysts of $S_F^\dagger$ for inputs in the two-dimensional space ${\cal K}\defeq\mathrm{span}\{\ket{\pi_M},\ket{\pi_{\overline M}}\}$. Accordingly, define
\begin{equation}\label{eq:W-K}
    W(S_F^\dagger,{\cal K})\defeq\sup\left\{W(S_F^\dagger,\ket y):\ket y\in{\cal K},\ \norm{\ket y}=1\right\}.
\end{equation}
Note that $W(S_F^\dagger,{\cal K})$ can be much smaller than $W(S_F)$, which is a supremum over the whole public space of $S_F$.

\begin{theorem}\label{thm:aa-canonical}
Let $A'=U_F A$ for unitaries $A$ and $U_F$. Suppose $S_F$ is a transducer for $U_F$ with public space ${\cal V}$ and private space ${\cal L}_F$, let $W(S_F^\dagger,{\cal K})$ be as in \eqref{eq:W-K}, and let $\eta\in (0,1)$. There exists a positive integer $D=O(1/(\eta\sqrt{\mu}))$ and a transducer $S_{\rm AA}^{[D]}(A)$, canonical with respect to the oracle that applies $A$ or $A^\dagger$, on ${\cal F}\otimes {\cal C}^{[D]}\otimes {\cal T}\otimes ({\cal V} \oplus {\cal L}_F)$, where ${\cal F}$ and ${\cal C}^{[D]}$ are as in \Cref{thm:aa-truncated-intro}, and ${\cal T}=\mathrm{span}\{\ket{0},\dots,\ket{4}\}$, such that:
\begin{enumerate}
    \item its public space is $(\mathrm{span}\{\ket{0}\ket{0}\ket{0}\}\otimes {\cal V})\oplus (\mathrm{span}\{\ket{1}\}\otimes {\cal C}^{[D]}\otimes {\cal T}\otimes {\cal V})$;
    \item whenever $\mu>0$, it has transduction action $S_{\rm AA}^{[D]}(A):\ket{0}\ket{0}\ket{0}\ket{\pi}\rightsquigarrow\ket{\tau}_{{\cal FC}^{[D]}{\cal TV}}$ such that $\norm{\ket{\tau}-\ket{1}\ket{\lambda'}\ket{0}\ket{\pi_M}}\leq \eta$ for some unit vector $\ket{\lambda'}\in \mathbb{C}^{[D]}$ that depends only on $\mu$ and $D$;
    \item $W(S_{\rm AA}^{[D]}{(A)},\ket{0}\ket{0}\ket{0}\ket{\pi})=O((1+W(S_F^\dagger,{\cal K}))/\sqrt{\mu})$ and $L(S_{\rm AA}^{[D]}(A),A,\ket{0}\ket{0}\ket{0}\ket{\pi})=O(1/\sqrt{\mu})$;
    \item One application of the work unitary of $S_{\rm AA}^{[D]}(A)$ uses one controlled call to each of $S_F$ and $S_F^\dagger$, $O(1)$ controlled calls to $R$, and $O(\log \frac{1}{\eta\sqrt{\mu}}+\log\dim{\cal V})$ additional one- and two-qubit gates.
\end{enumerate}
\end{theorem}
\noindent This theorem is restated in \Cref{sec:aa-transducer-second} as \Cref{thm:aa-canonical-truncated}.

The construction in \Cref{sec:aa-transducer-second} is very similar to {the one in} \Cref{sec:aa-transducer-first}, and while the main result of this paper does not depend on the construction in \Cref{sec:aa-transducer-first}, it is a good warm-up, and may also be of independent interest.

\subsubsection{Transducer for amplitude amplification: general case}\label{sec:aa-transducer-first}

Our transducer for amplitude amplification is comparable to a quantum algorithm for amplitude amplification that peels off some of the marked state at every step, before applying the next Grover iterate. Such an algorithm was studied, for example, in~\cite{Mizel2008dampedSearch}.

\begin{figure}[t]
    \centering
\begin{tikzpicture}[x=2.35cm,y=1.15cm,
    site/.style={circle,fill=black,inner sep=1.4pt},
    lab/.style={font=\small}]
    \node (e0) at (-.5,0) {};
    \node[site] (s0) at (0,0) {};
    \node (e1) at (.5,0) {};
    \node[site] (s1) at (1,0) {};
    \node (e2) at (1.5,0) {};
    \node[site] (s2) at (2,0) {};
    \node (e3) at (2.5,0) {};
    \node[site] (s3) at (3,0) {};
    \node[lab] at (4.5,0) {$\cdots$};
    \draw[->] (-1,0)--(s0)
    node[midway,below,lab] {$\ket{0}_{\cal F}\ket{0}_{\cal C}\ket{\pi}$};
    \draw[->] (s0)--(s1)
    node[midway,below,lab] {$\ket{0}_{\cal F}\ket{1}_{\cal C}\ket{\pi_1}$};
    \draw[->] (s1)--(s2) node[midway,below,lab] {$\ket{0}_{\cal F}\ket{2}_{\cal C}\ket{\pi_2}$};
    \draw[->] (s2)--(s3) node[midway,below,lab] {$\ket{0}_{\cal F}\ket{3}_{\cal C}\ket{\pi_3}$};
    \draw[->] (s3)--(4,0)
    node[midway,below,lab] {$\ket{0}_{\cal F}\ket{4}_{\cal C}\ket{\pi_4}$};
    \draw[->] (s0)--(0,2)
    node[midway,rotate=90,above,lab] {$\lambda_0\ket{1}_{\cal F}\ket{0}_{\cal C}\ket{\pi_M}$};
    \draw[->] (s1)--(1,2)
    node[midway,rotate=90,above,lab] {$\lambda_1\ket{1}_{\cal F}\ket{1}_{\cal C}\ket{\pi_M}$};
    \draw[->] (s2)--(2,2)
    node[midway,rotate=90,above,lab] {$\lambda_2\ket{1}_{\cal F}\ket{2}_{\cal C}\ket{\pi_M}$};
    \draw[->] (s3)--(3,2)
    node[midway,rotate=90,above,lab] {$\lambda_3\ket{1}_{\cal F}\ket{3}_{\cal C}\ket{\pi_M}$};
\end{tikzpicture}
    \caption{Visualization of the amplitude amplification transducer, which can be implemented as a quantum walk on this infinite comb. We implement it more directly by the process in \Cref{fig:amp-amp-transducer}.}
    \label{fig:aa-comb}
\end{figure}

Our transducer can be implemented as a quantum walk on a graph like the one in \Cref{fig:aa-comb}, or as the process described in \Cref{fig:amp-amp-transducer}, which can be visualized similarly -- we will analyze this process. Specifically, the transducer acts on ${\cal F}\otimes {\cal C}\otimes {\cal V}$, where ${\cal F}$ is a single-qubit flag register, and ${\cal C}=\mathrm{span}\{\ket{\ell}:\ell\in\mathbb{N}_0\}$ is an infinite counter\footnote{Incrementing is an isometry rather than a unitary, but this is no problem in the transducer framework \cite{JW:optQPE} (and eventually we truncate to finite dimensions and a unitary cyclic counter).}. In one step of the transducer, it peels off a small piece of the marked part of the state via the following \emph{leak unitary}, $W_{\rm leak}({a})${ for $a\in[0,1]$}, which sets the flag to 1 (done) on a portion of the marked part of the state:
\begin{equation}\label{eq:W-leak}
    W_{\rm leak}({a})\ket{0} = {a}\ket{0}+{\sqrt{1-a^2}}\ket{1}.
\end{equation}
If we apply this to any state $\ket{0}\ket{\psi}=\ket{0}(\alpha\ket{\pi_M}+\beta\ket{\pi_{\overline{M}}})$ controlled on being in the marked subspace, we get:
\begin{equation}\label{eq:controlled-leak-action}
\begin{split}
    \ket{0}\ket{\psi}\mapsto{}& \alpha W_{\rm leak}(a)\ket{0}\ket{\pi_M}+\beta\ket{0}\ket{\pi_{\overline{M}}}\\
    ={}& {a}\alpha\ket{0}\ket{\pi_M}+{\sqrt{1-a^2}}\alpha\ket{1}\ket{\pi_M}+\beta\ket{0}\ket{\pi_{\overline{M}}}\\
    ={}& \ket{0}({a}\Pi_M+\Pi_{\overline{M}})\ket{\psi}+{\sqrt{1-a^2}}\alpha\ket{1}\ket{\pi_M}.
\end{split}
\end{equation}
The precise amount of leakage depends on the value of the counter, via the \emph{block index} $j(\ell)\defeq\lfloor\log_4\ell\rfloor$ for $\ell\geq1$, and $j(0)\defeq0$. We use the schedule
\begin{equation}\label{eq:schedule}
    a_\ell\defeq\frac{1-4^{-j(\ell)}}{1+4^{-j(\ell)}},\qquad\mbox{so that}\qquad \sqrt{1-a_\ell^2}=\frac{2\cdot 2^{-j(\ell)}}{1+4^{-j(\ell)}}.
\end{equation}
The schedule is constant on each block $B_j\defeq\{4^j,\dots,4^{j+1}-1\}$. Note that $a_\ell=0$ for $\ell\leq3$, so during the first four steps \emph{all} marked amplitude leaks. For $\ell\in B_j$ we have $4^{-j}\in[\frac1\ell,\frac4\ell)$, and hence $1-a_\ell=\frac{2\cdot4^{-j}}{1+4^{-j}}\in[\frac1\ell,\frac8\ell)$: roughly a $\Theta(1/\ell)$ fraction of the marked amplitude leaks at step $\ell$.

\begin{remark}\label{rem:old-schedule}
A smoother schedule of $a_\ell=\sqrt{\max\{1-5/\ell,0\}}$ is also possible. \Cref{lem:S_AA-complexity} also holds for that schedule, with a similar proof.
The main drawback of that schedule is the implementation cost. The required reversible fixed-point arithmetic costs $O(\log^2D)$ gates per step for a ${\log D}$-bit number. The schedule \eqref{eq:schedule} needs only $O(\log D)$ gates, and is implemented exactly (see \Cref{thm:aa-truncated}).
\end{remark}

After leaking a portion of the marked part of the state, one Grover iterate, defined
\begin{equation}\label{eq:grover-iterate}
    G\defeq\big(2\ketbra\pi\pi-I\big)\big(I-2\Pi_M\big)=A'(I-2\ketbra00){A'}^\dagger {R},
\end{equation}
is applied to slightly boost the marked part of the state again, and the counter is incremented. Intuitively, we can think of applying this process repeatedly, and measuring the flag register each time. If we ever measure a 1, we know we have the marked state, otherwise we can continue. The number of steps we take is potentially unbounded, but finite in expectation, as we will essentially see in our analysis of the transducer.

\begin{figure}
\hrule\vspace{3pt}
\noindent\textbf{The transducer $S_{\rm AA}$ for amplitude amplification}
\vspace{3pt}\hrule\vspace{1.5pt}\hrule\vspace{6pt}
\noindent \textbf{Public Space:} ${\cal H}=\mathrm{span}\{\ket{0}_{\cal F}\ket{0}_{\cal C}\}\otimes{\cal V}\oplus \mathrm{span}\{\ket{1}_{\cal F}\}\otimes {\cal C}\otimes {\cal V}$\\
\textbf{Private Space:} ${\cal L}=\mathrm{span}\{\ket{0}_{\cal F}\ket{\ell}_{\cal C}:\ell\geq 1\}\otimes{\cal V}$
\vspace{4pt}\hrule
\begin{enumerate}
    \item Controlled on ${\cal V}$ in the marked space, using the value $\ell$ in ${\cal C}$, apply $W_{\rm leak}({a_\ell})$ to ${\cal F}$ (see \Cref{eq:W-leak}).
    \item Controlled {on} ${\cal F}=0$, apply the Grover iterate $G=A'(I-2\ket{0}\bra{0}){A'}^\dagger R$ to ${\cal V}$.
    \item Controlled on ${\cal F}=0$, increment ${\cal C}$.
\end{enumerate}
\hrule
\caption{We can see by inspection that the transducer uses one controlled call to each of $A'$ and ${A'}^\dagger$ (Step 2), and $O(1)$ controlled calls to $R$ (Steps 1 and 2).}\label{fig:amp-amp-transducer}
\end{figure}

Define a series of vectors representing the remaining unflagged state after $\ell$ applications of the process:
\begin{equation}\label{eq:pi-ell}\ket{\pi_0}\defeq\ket{\pi},
\quad\mbox{ and for all }\ell\geq 0,
\quad
\ket{\pi_{\ell+1}} \defeq G ({a_{\ell}}\Pi_M+\Pi_{\overline{M}})\ket{\pi_{\ell}},
\end{equation}
where $G$ is as in \Cref{eq:grover-iterate}. Note that while $\ket{\pi_0}$ is a unit vector, $\norm{\ket{\pi_{\ell}}}$ will generally be strictly smaller than $\norm{\ket{\pi_{\ell-1}}}$.
Let $\alpha^\ell\defeq\braket{\pi_M}{\pi_{\ell}}$ and $\beta^\ell\defeq\braket{\pi_{\overline{M}}}{\pi_{\ell}}$ so that
$$\ket{\pi_{\ell}}=\alpha^{\ell}\ket{\pi_M}+{\beta^{\ell}}\ket{\pi_{\overline{M}}}.$$
Since $\alpha^0=\sqrt\mu$ and $\beta^0=\sqrt{1-\mu}$ are real and $G$ acts as a real rotation on $\mathrm{span}\{\ket{\pi_M},\ket{\pi_{\overline M}}\}$, all $\alpha^\ell,\beta^\ell$ are real.
Define, for $\ell\geq 0$:
$$\lambda_\ell\defeq \alpha^{\ell}{\sqrt{1-a_{\ell}^2}},$$
which represents the amount of marked amplitude leaked at step $\ell$. More concretely, the transducer $S_{\rm AA}$ acts as follows on the state $\ket{0}_{\cal F}\ket{\ell}_{\cal C}\ket{\pi_{\ell}}_{\cal V}$:
\begin{equation}
\begin{split}
    S_{\rm AA}:\ket{0}_{\cal F}\ket{\ell}_{\cal C}\ket{\pi_{\ell}}_{\cal V} &\overset{{\rm Step 1}}{\longmapsto}
    \ket{0}_{\cal F}\ket{\ell}_{\cal C}({a_{\ell}}\Pi_M+\Pi_{\overline{M}})\ket{\pi_{\ell}}
    +\alpha^\ell{\sqrt{1-a_{\ell}^2}}\ket{1}_{\cal F}\ket{\ell}_{\cal C}\ket{\pi_M} \qquad \mbox{by \eqref{eq:controlled-leak-action}}\\
    &= \ket{0}_{\cal F}\ket{\ell}_{\cal C}({a_{\ell}}\Pi_M+\Pi_{\overline{M}})\ket{\pi_{\ell}}+\lambda_\ell\ket{1}_{\cal F}\ket{\ell}_{\cal C}\ket{\pi_M}\\
    &\overset{{\rm Step 2}}{\longmapsto}
    \ket{0}_{\cal F}\ket{\ell}_{\cal C}G({a_{\ell}}\Pi_M+\Pi_{\overline{M}})\ket{\pi_{\ell}}+\lambda_\ell\ket{1}_{\cal F}\ket{\ell}_{\cal C}\ket{\pi_M}\\
    &=\ket{0}_{\cal F}\ket{\ell}_{\cal C} \ket{\pi_{\ell+1}}+\lambda_\ell\ket{1}_{\cal F}\ket{\ell}_{\cal C}\ket{\pi_M}\\
    &\overset{{\rm Step 3}}{\longmapsto}
    \ket{0}_{\cal F}\ket{\ell+1}_{\cal C} \ket{\pi_{\ell+1}}+\lambda_\ell\ket{1}_{\cal F}\ket{\ell}_{\cal C}\ket{\pi_M}.
\end{split}\label{eq:S_AA-action}
\end{equation}

To see that the transducer has the correct action, we define the catalyst:
\begin{equation}\label{eq:S_AA-catalyst}
\ket{v}\defeq \sum_{\ell=1}^\infty \ket{0}_{\cal F}\ket{\ell}_{\cal C}\ket{\pi_{\ell}},
\end{equation}
which converges as long as $\mu>0$:
\begin{lemma}\label{lem:S_AA-complexity}
Let $\ket{v}$ be as in \Cref{eq:S_AA-catalyst}. Then $\norm{\ket{v}}^2=\sum_{\ell=1}^\infty\norm{\ket{\pi_{\ell}}}^2=O(1/\sqrt{\mu})$.
\end{lemma}
The proof of \Cref{lem:S_AA-complexity} is given shortly. We first prove that $\ket{v}$ is in fact a catalyst for the claimed action of $S_{\rm AA}$.
\begin{lemma}\label{lem:S_AA-action} Let $\ket{\lambda}\defeq \sum_{\ell=0}^\infty\lambda_\ell\ket{\ell}$. Then
    $$S_{\rm AA}(\ket{0}_{\cal F}\ket{0}_{\cal C}\ket{\pi}_{\cal V} +\ket{v}) =\ket{1}_{\cal F}\ket{\lambda}_{\cal C}\ket{\pi_M}_{\cal V} +\ket{v}.$$
\end{lemma}
\begin{proof}
A simple calculation shows:
    \begin{align*}
        \ket{0}_{\cal F}\ket{0}_{\cal C}\ket{\pi}_{\cal V} +\ket{v}
        &= \sum_{\ell=0}^{\infty}\ket{0}_{\cal F}\ket{\ell}_{\cal C}\ket{\pi_{\ell}}_{\cal V}\\
        &\overset{S_{\rm AA}}{\mapsto}
        \sum_{\ell=0}^{\infty}\left(\ket{0}\ket{\ell+1}\ket{\pi_{\ell+1}}+\lambda_\ell\ket{1}\ket{\ell}\ket{\pi_M}\right)
        & \mbox{by \eqref{eq:S_AA-action}}\\
        &= \sum_{\ell=1}^{\infty}\ket{0}\ket{\ell}\ket{\pi_\ell}+\sum_{\ell=0}^\infty\lambda_\ell\ket{1}\ket{\ell}\ket{\pi_M}
        =\ket{v}+\ket{1}\ket{\lambda}\ket{\pi_M}.& \qedhere
    \end{align*}
\end{proof}

\begin{remark}\label{rem:lambda-unit}
    By unitarity of $W_{\rm leak}$ and $G$, $\norm{\ket{\pi_\ell}}^2=\norm{\ket{\pi_{\ell+1}}}^2+\lambda_\ell^2$ for every $\ell$, and \Cref{lem:S_AA-complexity} gives $\norm{\ket{\pi_\ell}}\to0$. Hence $\sum_{\ell\geq0}\lambda_\ell^2=\norm{\ket{\pi_0}}^2=1$: the output $\ket1\ket\lambda\ket{\pi_M}$ is a unit vector, as it must be for a transduction action. The vector $\ket\lambda$ depends only on $\mu$.
\end{remark}

\noindent We now prove \Cref{lem:S_AA-complexity}, upper bounding the catalyst size, $\norm{\ket{v}}^2$.

\begin{proof}[Proof of \Cref{lem:S_AA-complexity}]
Throughout this proof, we interpret the vectors $\ket{\pi_\ell}$ as real 2-dimensional vectors in $\mathrm{span}\{\ket{\pi_M},\ket{\pi_{\overline{M}}}\}$. The Grover iterate is well known to map this space to itself, acting on it as
$$R_\theta=\begin{pmatrix}\cos\theta&\sin\theta\\-\sin\theta&\cos\theta\end{pmatrix},
\quad\mbox{where}\quad \sin\frac\theta2=\sqrt\mu,\ \theta\in(0,\pi/2],$$
using $\mu\leq\frac12$. By the double angle identities, $\sin\theta=2\sqrt{\mu(1-\mu)}$ and $\cos\theta=1-2\mu\in[0,1)$. Since $1-2\mu\leq\sqrt{1-\mu}$,
\begin{equation}\label{eq:cot-bound}
    0\leq\cot\theta=\frac{\cos\theta}{\sin\theta}=\frac{1-2\mu}{2\sqrt{\mu(1-\mu)}}\leq\frac{1}{2\sqrt\mu}.
\end{equation}
Note that for any $T=O(1/\sqrt{\mu})$,
$$\sum_{\ell=1}^T\norm{\ket{\pi_\ell}}^2=O(1/\sqrt{\mu}),$$
since for all $\ell\geq 1$, $\norm{\ket{\pi_\ell}}^2\leq \norm{\ket{\pi_0}}^2=1$. Thus, to upper bound the infinite sum, we only need to show that $\norm{\ket{\pi_\ell}}^2$ decreases sufficiently quickly in the $\ell>T$ regime. We will instead define a measure $q_\ell$ based on an alternative norm of $\ket{\pi_\ell}$, and show that this is a good approximation of $\norm{\ket{\pi_\ell}}^2$, and it decreases quickly.

To this end, define, for $a\in[0,1]$,
$$M(a)\defeq R_\theta\begin{pmatrix}a&0\\0&1\end{pmatrix},
\qquad
P(a)\defeq\begin{pmatrix}a&b(a)\\b(a)&1\end{pmatrix},
\quad\mbox{where}\quad b(a)\defeq\frac{(a-1)\cot\theta}{2}.$$
By \eqref{eq:pi-ell}, $\ket{\pi_{\ell+1}}=M(a_\ell)\ket{\pi_\ell}$, and a direct calculation gives
\begin{equation}\label{eq:P-relation}
    \forall a\in[0,1],\quad M(a)^TP(a)M(a)=aP(a).
\end{equation}
Let $C\defeq40$, let $J$ be the smallest integer with $4^J\geq C/\sqrt\mu$, and let $T\defeq4^J$. Since $C/\sqrt\mu>1$, we have $J\geq1$ and $T<4C/\sqrt\mu=160/\sqrt\mu$. For $j\geq J$, write $s_j\defeq4^{-j}\leq\sqrt\mu/C$, let $a^{(j)}$ be the common value of $a_\ell$ for $\ell\in B_j$, and let $P_j\defeq P(a^{(j)})$. For $\ell\geq T$, let $q_\ell\defeq\bra{\pi_\ell}P_{j(\ell)}\ket{\pi_\ell}$.

\begin{claim}\label{claim:P-bounds}
For all $j\geq J$, $\frac14I\leq P_j\leq 2I$ and $\norm{P_{j+1}-P_j}\leq\frac{5}{2C}$, and thus, for all $\ell\geq T$ $\frac14\norm{\ket{\pi_\ell}}^2\leq q_\ell\leq2\norm{\ket{\pi_\ell}}^2$.
\end{claim}
\begin{proof}
We have $1-a^{(j)}=\frac{2s_j}{1+s_j}\leq 2s_j$, so by \eqref{eq:cot-bound}, $\abs{b(a^{(j)})}\leq s_j\cot\theta\leq\frac{s_j}{2\sqrt\mu}\leq\frac{1}{2C}$. The smallest eigenvalue of a symmetric $2\times2$ matrix of the form $P(a)$ is at least $\min\{a,1\}-\abs{b(a)}$, so it is at least $1-2s_j-\frac1{2C}\geq1-\frac{5}{2C}\geq\frac14$, using $s_j\leq\frac1C$. Since $\mathrm{Tr}(P_j)=a^{(j)}+1\leq2$ and $P_j\geq0$, we also have $P_j\leq 2I$. For the second statement,
$$P_{j+1}-P_j=\big(a^{(j+1)}-a^{(j)}\big)\begin{pmatrix}1&\frac{\cot\theta}{2}\\\frac{\cot\theta}{2}&0\end{pmatrix},$$
with $0\leq a^{(j+1)}-a^{(j)}\leq1-a^{(j)}\leq2s_j$. The matrix in the above expression has norm at most $1+\frac{\cot\theta}{2}\leq1+\frac{1}{4\sqrt\mu}$, so $\norm{P_{j+1}-P_j}\leq 2s_j+\frac{s_j}{2\sqrt\mu}\leq\frac2C+\frac1{2C}=\frac{5}{2C}$.
\end{proof}

The previous claim shows that $q_\ell$ is a reasonably good proxy for $\norm{\ket{\pi_\ell}}^2$. Next we show that it decreases quickly.

\begin{claim}\label{claim:q-step}
Let $\ell\geq T$. If $\ell$ and $\ell+1$ are in the same block, then $q_{\ell+1}=a_\ell q_\ell$. Otherwise, $q_{\ell+1}\leq\left(1+\frac{10}{C}\right)a_\ell q_\ell$.
\end{claim}
\begin{proof}
Let $j=j(\ell)$. First suppose $j(\ell+1)=j$, so $\ell$ and $\ell+1$ are in the same block. Using $\ket{\pi_{\ell+1}}=M(a_{\ell})\ket{\pi_\ell}$ and $a_{\ell}=a_{\ell+1}=a^{(j)}$, we have:
$$q_{\ell+1}=\bra{\pi_{\ell+1}}P_j\ket{\pi_{\ell+1}}=\bra{\pi_\ell}M(a_\ell)^TP(a_\ell)M(a_\ell)\ket{\pi_\ell}=
a_\ell\bra{\pi_\ell}P(a_\ell)\ket{\pi_{\ell}}
=a_\ell q_\ell,$$
by \eqref{eq:P-relation}. This proves the first case. In the second case, $j(\ell+1)=j+1$, and by \Cref{claim:P-bounds},
$$q_{\ell+1}\leq\bra{\pi_{\ell+1}}P_j\ket{\pi_{\ell+1}}+\norm{P_{j+1}-P_j}\norm{\ket{\pi_{\ell+1}}}^2\leq a_\ell q_\ell+\frac{5}{2C}\cdot4\bra{\pi_{\ell+1}}P_j\ket{\pi_{\ell+1}}=\left(1+\frac{10}{C}\right)a_\ell q_\ell.\qedhere$$
\end{proof}

Now fix $j\geq J$. The block $B_j$ has $3\cdot4^j=3/s_j$ elements, and $a^{(j)}=1-\frac{2s_j}{1+s_j}\leq1-s_j$. Applying \Cref{claim:q-step} for all $\ell\in B_j$,
$$q_{4^{j+1}}\leq\left(1+\frac{10}{C}\right)(1-s_j)^{3/s_j}q_{4^j}\leq\frac{5}{4}e^{-3}q_{4^j}\leq\frac{1}{16}q_{4^j}.$$
Moreover, $q_\ell$ is non-increasing within each block. So for $\ell\in B_{J+m}$ with $m\geq0$, we have $q_\ell\leq16^{-m}q_T\leq 2\cdot16^{-m}$. Since $\ell<4^{J+m+1}=4^{m+1}T$, we have $16^{-m}<16(T/\ell)^2$, and therefore
\begin{equation}\label{eq:main-thing-to-prove}
    \forall\ell\geq T,\quad\norm{\ket{\pi_\ell}}^2\leq4q_\ell\leq128\left(\frac{T}{\ell}\right)^2.
\end{equation}
Using $\norm{\ket{\pi_\ell}}\leq1$ for $\ell<T$ and $\sum_{\ell\geq T}\ell^{-2}\leq\frac{1}{T^2}+\frac1T\leq\frac2T$, we conclude\footnote{This constant in the final bound is very loose, and it would be interesting to give sharper bounds (given that a large constant overhead could kill a log factor improvement in practice).}
$$\norm{\ket v}^2=\sum_{\ell=1}^{T-1}\norm{\ket{\pi_\ell}}^2+\sum_{\ell\geq T}\norm{\ket{\pi_\ell}}^2\leq T+256\,T=257\,T<\frac{257\cdot160}{\sqrt\mu}=O\left(\frac1{\sqrt\mu}\right).$$
\end{proof}

This concludes the proof of \Cref{thm:aa-transducer}. We finish this section by showing that it is possible to truncate the transducer construction of \Cref{fig:amp-amp-transducer} in such a way that it has an efficient implementation, thus establishing \Cref{thm:aa-truncated-intro}.
For any integer $D\geq 1$, let $S_{\rm AA}^{[D]}$ be obtained from $S_{\rm AA}$ by replacing ${\cal C}$ with
$${\cal C}^{[D]}=\mathrm{span}\{\ket{0},\dots,\ket{D-1}\},$$
and the incrementation in Step~3 by incrementation modulo $D$. Its public space is $$\left(\mathrm{span}\left\{\ket{0}_{\cal F}\ket{0}_{{\cal C}^{[D]}}\right\}\otimes{\cal V}\right)\oplus\left(\mathrm{span}\{\ket{1}_{\cal F}\}\otimes{{\cal C}^{[D]}}\otimes{\cal V}\right),$$
and its private space is
$$\mathrm{span}\left\{\ket{0}_{\cal F}\ket{\ell}_{{\cal C}^{[D]}}:1\leq\ell\leq D-1\right\}\otimes{\cal V}.$$
\begin{theorem}\label{thm:aa-truncated}
Let $S_{\rm AA}^{[D]}$ be the transducer described above.
Let $T<160/\sqrt\mu$ be as in the proof of \cref{lem:S_AA-complexity}. Then:
\begin{enumerate}
    \item $S_{\rm AA}^{[D]}$ transduces $\ket{0}_{\cal F}\ket{0}_{{\cal C}^{[D]}}\ket{\pi}$ into a state $\ket{\tau}$ such that
    $$\norm{\ket{\tau}-\ket{1}\ket{\lambda'}\ket{\pi_M}}\leq \frac{16\,T}{D},$$
    for some unit vector $\ket{\lambda'}\in{\cal C}^{[D]}$ that depends only on $\mu$ and $D$;
    \item $W(S_{\rm AA}^{[D]},\ket{0}\ket{0}\ket{\pi})=O(1/\sqrt{\mu})$;
    \item $S_{\rm AA}^{[D]}$ can be implemented using one controlled call to each of $A'$ and ${A'}^\dagger$, $O(1)$ controlled calls to $R$, and $O(\log D+\log\dim{\cal V})$ additional one- and two-qubit gates.
\end{enumerate}
In particular, for any $\eta\in(0,1)$, taking $D=\ceil{16T/\eta}=O(1/(\eta\sqrt{\mu}))$ makes the error in Item~1 at most $\eta$, and the counter register then has $\ceil{\log D}=O(\log(1/\mu)+\log(1/\eta))$ qubits.
\end{theorem}
\begin{proof}
We first bound the cost of the three steps in \Cref{fig:amp-amp-transducer}.
\paragraph{Step 1.} The counter value $\ell\in\{0,\dots,D-1\}$ is stored in $n\defeq\ceil{\log D}$ bits $\ell_{n-1}\dots\ell_0$. For $j=\ceil{n/2},\dots,1$, compute into ancillas the bits $z_j\defeq[\ell\geq4^j]=z_{j+1}\vee\ell_{2j+1}\vee\ell_{2j}$ (with $z_{\ceil{n/2}+1}\defeq0$ and $z_0\defeq1$), using $O(1)$ gates each. Then $e_j\defeq z_j\wedge\neg z_{j+1}$ equals $1$ if and only if $j(\ell)=j$. By \eqref{eq:schedule},
there are $\lceil n/2\rceil$ possibilities for $W_{\rm leak}(a_\ell)$, $W_1,\dots,W_{\ceil{n/2}}$, depending on to which block $\ell$ belongs. For each $j$, we apply the fixed single-qubit gate $W_j$ to ${\cal F}$, controlled on being in the marked subspace (checked via controlled calls to $R$), and $e_j=1$. Finally, uncompute the $z_j$. Exactly one $e_j$ is $1$, so this applies $W_{\rm leak}(a_\ell)$ exactly, controlled on the marked subspace, using $O(\log D)$ gates and $O(1)$ controlled calls of $R$.

\paragraph{Step 2.} This step is unchanged, and uses one controlled call to each of $A'$, ${A'}^\dagger$ and $R$, as well as $I-2\ket{0}\bra{0}$, also acting on ${\cal V}$, which can be implemented in $O(\log\dim{\cal V})$ gates.

\paragraph{Step 3.} The increment modulo $D$ can be implemented in $O(\log D)$ gates.

\bigskip

\noindent This proves Item~3.

\paragraph{Transduction action.} Let $\ket{\pi_\ell}$ and $\lambda_\ell$ be as in \cref{sec:aa-transducer-first}. Let
$$\ket{v^{[D]}}\defeq\sum_{\ell=1}^{D-1}\ket{0}_{\cal F}\ket{\ell}_{{\cal C}^{[D]}}\ket{\pi_\ell},$$
which lies in the private space. Exactly as in \eqref{eq:S_AA-action}, for each $\ell\in\{0,\dots,D-1\}$,
$$S_{\rm AA}^{[D]}:\ket{0}_{\cal F}\ket{\ell}\ket{\pi_\ell}\mapsto\ket{0}_{\cal F}\ket{\ell+1\bmod D}\ket{\pi_{\ell+1}}+\lambda_\ell\ket{1}_{\cal F}\ket{\ell}\ket{\pi_M}.$$
Summing over $\ell$ as in the proof of \cref{lem:S_AA-action}, the only difference is that the term with $\ell=D-1$ wraps around to $\ket{\ell+1 \mod D}_{{\cal C}^{[D]}}=\ket{0}_{{\cal C}^{[D]}}$:
$$S_{\rm AA}^{[D]}\left(\ket{0}\ket{0}\ket{\pi}+\ket{v^{[D]}}\right)=\underbrace{\ket{1}_{\cal F}\ket{\lambda^{[D]}}_{{\cal C}^{[D]}}\ket{\pi_M}+\ket{0}_{\cal F}\ket{0}_{{\cal C}^{[D]}}\ket{\pi_D}}_{\eqqcolon\ket{\tau}}+\ket{v^{[D]}},\;\mbox{where }\ket{\lambda^{[D]}}\defeq\sum_{\ell=0}^{D-1}\lambda_\ell\ket{\ell}.$$
The wrapped-around term $\ket{0}_{\cal F}\ket{0}_{{\cal C}^{[D]}}\ket{\pi_D}$ lies in the public space. So this is an exact transduction $\ket{0}_{\cal F}\ket{0}_{\cal C}\ket{\pi}\overset{S_{\rm AA}^{[D]}}{\rightsquigarrow}\ket{\tau}$ with catalyst $\ket{v^{[D]}}$, and $\norm{\ket{v^{[D]}}}^2\leq\norm{\ket v}^2=O(1/\sqrt\mu)$ by \cref{lem:S_AA-complexity}. This proves Item~2.

\paragraph{Error.} By \cref{rem:lambda-unit}, $\norm{\ket{\lambda^{[D]}}}^2=1-p$, where $p\defeq\norm{\ket{\pi_D}}^2<1$. Let $\ket{\lambda'}\defeq\ket{\lambda^{[D]}}/\sqrt{1-p}$, a unit vector that depends only on $\mu$ and $D$. Then
$$\norm{\ket{\tau}-\ket{1}\ket{\lambda'}\ket{\pi_M}}^2=\left(1-\sqrt{1-p}\right)^2+p\leq p^2+p\leq 2p.$$
If $D\geq T$, then \eqref{eq:main-thing-to-prove} gives $2p\leq 256(T/D)^2$, so the error is at most $16T/D$. If $D<T$, the bound holds trivially, since $1<16T/D$. This proves Item~1.
\end{proof}

\subsubsection{Canonical transducer in $A$}\label{sec:aa-transducer-second}

In this section, we use a standard technique for making a transducer canonical (see \cite[{Section~10}]{belovs2024taming}).
Suppose $A'=U_FA$ for some unitary $U_F$ (for us it will correspond to the filter LCU from \Cref{sec:lcu-transducer}), so that $\ket\pi=U_FA\ket0_{\cal V}$. Then we can break the Grover iterate $G=U_F A (I-2\ket{0}\bra{0}) A^\dagger U_F^\dagger R$ up into its parts, some of which are oracle calls to $A$. We will make a transducer for amplitude amplification that is similar to the transducer in \Cref{sec:aa-transducer-first}, except that it is canonical in $A$. We will also consider not necessarily having calls to $U_F$, but instead, calls to a transducer $S_F$ on ${\cal V}\oplus {\cal L}_F$, with public space ${\cal V}$, with transduction action $U_F$. This is strictly more general, with the special case being recovered when ${\cal L}_F=\{0\}$, in which case $S_F=U_F$. Define unitaries on ${\cal V}\oplus{\cal L}_F$:
\begin{equation}
    G_0=S_F^\dagger (R_{\cal V}\oplus I_{{\cal L}_F}),\; G_1=A^\dagger_{\cal V}\oplus I_{{\cal L}_F},\; G_2=(I-2\ket{0}\bra{0})_{\cal V}\oplus I_{{\cal L}_F},\; G_3=A_{\cal V}\oplus I_{{\cal L}_F}, \; G_4=S_F.
\end{equation}
In the special case $S_F=U_F$, $S_G\defeq G_4G_3G_2G_1G_0=G$. In the general case this is not true, but $S_G$ would be a transducer for $G$ if $G_0$ and $G_4$ were each applied to ${\cal V}$ and orthogonal copies of the private space ${\cal L}_F$ \cite[Propositions~9.1 and~9.9]{belovs2024taming}, which is essentially what happens in our construction below via the register ${\cal T}$. The proof of this is basically contained in \mbox{\Cref{lem:G-catalyst}}.

The canonical transducer will be the process in \Cref{fig:amp-amp-transducer-2}, which acts on one additional register, ${\cal T}=\mathrm{span}\{\ket0,\dots,\ket4\}$, which keeps a count of where we are in the process of applying the unitaries that make up $S_G$. Thus, instead of applying $S_G$ in one go, we break it up into parts, only applying the relevant part, and then incrementing the counter ${\cal T}$. In this way, in those applications of the transducer where we apply $A$ or $A^\dagger$, it's the first thing we apply, and it's not applied to the public space.

\begin{figure}
\hrule\vspace{3pt}
\noindent\textbf{The canonical transducer $S_{\rm AA}(A)$ for amplitude amplification}
\vspace{3pt}\hrule\vspace{1.5pt}\hrule\vspace{6pt}
\noindent\begin{tabular}{@{}l@{\ }l@{}}
\textbf{Public Space:} & ${\cal H}=(\mathrm{span}\{\ket{0}_{\cal F}\ket{0}_{\cal C}\ket{0}_{\cal T}\}\otimes{\cal V})\oplus (\mathrm{span}\{\ket{1}_{\cal F}\}\otimes {\cal C}\otimes {\cal T}\otimes {\cal V})$\\
\textbf{Private Space:} & ${\cal L}^{\bullet}=\mathrm{span}\{\ket{0}_{\cal F}\ket{\ell}_{\cal C}\ket{t}_{\cal T}:\ell\geq 0,t\in\{1,3\}\}\otimes{\cal V}$\\
 & ${\cal L}^{\circ}=(\mathrm{span}\{\ket{0}_{\cal F}\ket{\ell}_{\cal C}\ket{t}_{\cal T}:{(\ell\geq 1,t=0)\mbox{ or }(\ell\geq0,t\in\{2,4\})}\}\otimes{\cal V})$\\
 & $\phantom{{\cal L}^{\circ}}\oplus ({\cal F}\otimes {\cal C}\otimes {\cal T}\otimes {\cal L}_F)$
\end{tabular}
\vspace{4pt}\hrule
\begin{enumerate}
    \item Controlled on {being} in the marked subspace of ${\cal V}$ in the last register, if ${\cal T}=0$, use the value $\ell$ in ${\cal C}$ to apply $W_{\rm leak}({a_\ell})$ to ${\cal F}$ (see \Cref{eq:W-leak}).
    \item If ${\cal F}=0$: if ${\cal T}=1$ apply $A^\dagger$ to ${\cal V}$ and if ${\cal T}=3$ apply $A$ to ${\cal V}$.
    \item If ${\cal F}=0$: if ${\cal T}=t$ for $t\in\{0,2,4\}$, apply $G_t$ to ${\cal V}\oplus {\cal L}_F$.
    \item If ${\cal F}={\cal T}=0$, controlled on being in ${\cal V}$ in the last register, increment ${\cal C}$.
    \item If ${\cal F}=0$, controlled on being in ${\cal V}$ in the last register, increment ${\cal T}$ modulo 5.
\end{enumerate}
\hrule
\caption{The canonical transducer $S_{\rm AA}(A)$. We can see that it is canonical by noting that steps 1 and 2 commute: Step 1 acts non-trivially only on states with ${\cal T}=0$, and Step 2 only on states with ${\cal T}\in\{1,3\}$. We could have just as well combined Steps 2 and 3 into a single step that controls on ${\cal F}=0$ and applies $G_t$ where $t$ is the value in ${\cal T}$, but we wrote it this way to make it clearer it's canonical.}\label{fig:amp-amp-transducer-2}
\end{figure}

\paragraph{Catalyst.} Towards defining a catalyst, let $\ket{\pi_\ell}$ be as in \Cref{eq:pi-ell}, and define:
\begin{equation}
    \ket{\pi_\ell^0}\defeq \ket{\pi_\ell},
    \;\;
    \ket{\pi_\ell^1}\defeq U_F^\dagger {R}(a_{\ell}\Pi_M+\Pi_{\overline{M}})\ket{\pi_\ell},
    \;\;\mbox{for }t\in\{2,3,4\},\,
    \ket{\pi_\ell^t}\defeq G_{t-1}\ket{\pi_{\ell}^{t-1}},
    \;\;
    \ket{\pi_\ell^5}\defeq U_F\ket{\pi_{\ell}^4}.
\end{equation}
Note that we have
$$\ket{\pi_{\ell}^5}=G(a_{\ell}\Pi_M+\Pi_{\overline{M}})\ket{\pi_\ell}=\ket{\pi_{\ell+1}}.$$

\begin{lemma}\label{lem:G-catalyst}
    For any $\ell\geq 0$, there is a catalyst $\ket{v^{\ell}}\in {\cal F}\otimes {\cal C}\otimes \mathrm{span}\{\ket{0}_{\cal T},\ket{4}_{\cal T}\}\otimes {\cal L}_F$ such that $\norm{\ket{v^{\ell}}}^2\leq {2}\norm{\ket{\pi_{\ell+1}}}^2W(S_F^\dagger,{\cal K})$, and
    \begin{align*}
        &S_{\rm AA}(A)\left(\ket{0}_{\cal F}\ket{\ell}_{\cal C}\ket{0}_{\cal T}\ket{\pi_\ell}+\ket{0}_{\cal F}\ket{\ell+1}_{\cal C}\sum_{t=1}^4\ket{t}_{\cal T}\ket{\pi_{\ell}^t}_{\cal V}+\ket{v^{\ell}}\right)\\
        ={}&\ket{0}_{\cal F}\ket{\ell+1}_{\cal C}\sum_{t=1}^5\ket{t}_{\cal T}\ket{\pi_{\ell}^t}+\lambda_\ell\ket{1}_{\cal F}\ket{\ell}_{\cal C}\ket{0}_{\cal T}\ket{\pi_M}+\ket{v^{\ell}},
    \end{align*}
    {where we identify $\ket5_{\cal T}\ket{\pi_\ell^5}$ with $\ket0_{\cal T}\ket{\pi_{\ell+1}}$.}
    Moreover, for $\ell\neq \ell'$, $\braket{v^{\ell}}{v^{\ell'}}=0$.
\end{lemma}
\begin{proof}
Let $\ket{y_\ell}\defeq R(a_\ell\Pi_M+\Pi_{\overline M})\ket{\pi_\ell}$. Since $\ket{\pi_\ell}\in{\cal K}=\mathrm{span}\{\ket{\pi_M},\ket{\pi_{\overline{M}}}\}$, and both $R$ and $a_\ell\Pi_M+\Pi_{\overline M}$ map ${\cal K}$ to itself, we have $\ket{y_\ell}\in{\cal K}$. Moreover, $\norm{\ket{y_\ell}}=\norm{(a_\ell\Pi_M+\Pi_{\overline M})\ket{\pi_\ell}}=\norm{G(a_\ell\Pi_M+\Pi_{\overline M})\ket{\pi_\ell}}=\norm{\ket{\pi_{\ell+1}}}$, since $R$ and $G$ are unitary. By \eqref{eq:W-K}, there is a catalyst $\ket{v_\ell}\in{\cal L}_F$ for $\ket{y_\ell}$ in $S_F^\dagger$ with $\norm{\ket{v_\ell}}^2\leq\norm{\ket{\pi_{\ell+1}}}^2W(S_F^\dagger,{\cal K})$. Since $U_F^\dagger\ket{y_\ell}=\ket{\pi_\ell^1}$, this means that:
\begin{equation}\label{eq:S_F-cat}
\begin{split}
S_F^\dagger(\ket{y_{\ell}}+\ket{v_\ell}) &=U_F^\dagger\ket{y_\ell}+\ket{v_\ell}=\ket{\pi_\ell^1}+\ket{v_\ell}.
\end{split}
\end{equation}
Similarly, since $\ket{\pi_{\ell+1}}\in{\cal K}$, there is a catalyst $\ket{v_{\ell}'}\in {\cal L}_F$ for $\ket{\pi_{\ell+1}}$ in $S_F^\dagger$ with catalyst size $\norm{\ket{v'_{\ell}}}^2\leq\norm{\ket{\pi_{\ell+1}}}^2W(S_F^\dagger,{\cal K})$. Since $U_F^\dagger\ket{\pi_{\ell+1}}=U_F^\dagger\ket{\pi_\ell^5}=\ket{\pi_\ell^4}$, this means that:
\begin{equation}\label{eq:S_F-cat2}
\begin{split}
    S_F^\dagger (\ket{\pi_{\ell+1}}+\ket{v_\ell'}) &= U_F^\dagger\ket{\pi_{\ell+1}}+\ket{v_\ell'}\\
    \ket{\pi_{\ell+1}}+\ket{v_\ell'} &= S_F(\ket{\pi_\ell^4}+\ket{v_\ell'}).
\end{split}
\end{equation}
Then define:
\begin{equation}\label{eq:v-sup-ell}
\ket{v^\ell}:=\ket{0}_{\cal F}\ket{\ell}_{\cal C}\ket{0}_{\cal T}\ket{v_{\ell}}_{{\cal L}_F}+\ket{0}_{\cal F}\ket{\ell+1}_{\cal C}\ket{4}_{\cal T}\ket{v_\ell'}_{{\cal L}_F}
\end{equation}
so that
$\norm{\ket{v^\ell}}^2\leq 2\norm{\ket{\pi_{\ell+1}}}^2W(S_F^\dagger,{\cal K})$
(and $\braket{v^\ell}{v^{\ell'}}=0$ for $\ell\neq\ell'$, since they are supported on disjoint sets of values $(\ell,0),(\ell+1,4)$ of ${\cal C}\ot{\cal T}$)
and
\begin{equation}\label{eq:coupling}
\begin{split}
    &\ket{0}_{\cal F}\ket{\ell}_{\cal C}\ket{0}_{\cal T}\ket{\pi_\ell}+\ket{0}_{\cal F}\ket{\ell+1}_{\cal C}\sum_{t=1}^4\ket{t}_{\cal T}\ket{\pi_{\ell}^t}_{\cal V}+\ket{v^{\ell}}\\
    ={}& \ket{0}\ket{\ell}\ket{0}(\ket{\pi_\ell}+\ket{v_\ell})+\ket{0}\ket{\ell+1}\sum_{t=1}^3\ket{t}\ket{\pi_\ell^t}+\ket{0}\ket{\ell+1}\ket{4}(\ket{\pi_\ell^4}+\ket{v_\ell'}).
\end{split}
\end{equation}

The action of $S_{\rm AA}(A)$ on this vector depends on the value $t$, so we consider it in cases. First, for the $t=0$ part of the state, we can see the action step by step. The first step controls on being in the marked subspace of ${\cal V}$ in the last register, which does nothing to states in ${\cal L}_F$ in the last register, and acts on states with ${\cal V}$ in the last register as in \Cref{eq:controlled-leak-action}:
\begin{align*}
    &\ket{0}_{\cal F}\ket{\ell}_{\cal C}\ket{0}_{\cal T}(\ket{\pi_{\ell}}_{\cal V}+\ket{v_{\ell}}_{{\cal L}_F})\\
    \overset{1}{\mapsto}{}& \ket{0}_{\cal F}\ket{\ell}_{\cal C}\ket{0}_{\cal T}({a_{\ell}}\Pi_M+\Pi_{\overline{M}})\ket{\pi_\ell}_{\cal V}+{\sqrt{1-a_\ell^2}}\alpha^\ell\ket{1}_{\cal F}\ket{\ell}_{\cal C}\ket{0}_{\cal T}\ket{\pi_M}_{\cal V}+
    \ket{0}_{\cal F}\ket{\ell}_{\cal C}\ket{0}_{\cal T}\ket{v_{\ell}}_{{\cal L}_F}\\
    ={}& \ket{0}_{\cal F}\ket{\ell}_{\cal C}\ket{0}_{\cal T}(({a_{\ell}}\Pi_M+\Pi_{\overline{M}})\ket{\pi_\ell}_{\cal V}+\ket{v_\ell}_{{\cal L}_F})+\lambda_\ell\ket{1}_{\cal F}\ket{\ell}_{\cal C}\ket{0}_{\cal T}\ket{\pi_M}_{\cal V}\\
    \overset{2\& 3}{\mapsto}{}& \ket{0}_{\cal F}\ket{\ell}_{\cal C}\ket{0}_{\cal T} G_0(({a_{\ell}}\Pi_M+\Pi_{\overline{M}})\ket{\pi_\ell}_{\cal V}+\ket{v_\ell}_{{\cal L}_F})+\lambda_\ell\ket{1}_{\cal F}\ket{\ell}_{\cal C}\ket{0}_{\cal T}\ket{\pi_M}_{\cal V}\\
    ={}& \ket{0}_{\cal F}\ket{\ell}_{\cal C}\ket{0}_{\cal T} S_F^\dagger({R}({a_{\ell}}\Pi_M+\Pi_{\overline{M}})\ket{\pi_\ell}_{\cal V}+\ket{v_\ell}_{{\cal L}_F})+\lambda_\ell\ket{1}_{\cal F}\ket{\ell}_{\cal C}\ket{0}_{\cal T}\ket{\pi_M}_{\cal V}\\
    ={}& \ket{0}_{\cal F}\ket{\ell}_{\cal C}\ket{0}_{\cal T}(\ket{\pi_\ell^1}+\ket{v_\ell}) +\lambda_\ell\ket{1}_{\cal F}\ket{\ell}_{\cal C}\ket{0}_{\cal T}\ket{\pi_M}_{\cal V}
\end{align*}
by \Cref{eq:S_F-cat}, since steps 2 and 3, combined, simply apply $G_t$ controlled on $t$, conditioned on the flag being 0. Finally the last two steps act as
\begin{align*}
    \overset{4}{\mapsto}{}& \ket{0}_{\cal F}\ket{\ell+1}_{\cal C}\ket{0}_{\cal T}\ket{\pi_\ell^1}+\ket{0}_{\cal F}\ket{\ell}_{\cal C}\ket{0}_{\cal T}\ket{v_\ell} +\lambda_\ell\ket{1}_{\cal F}\ket{\ell}_{\cal C}\ket{0}_{\cal T}\ket{\pi_M}_{\cal V}\\
    \overset{5}{\mapsto}{}& \ket{0}_{\cal F}\ket{\ell+1}_{\cal C}\ket{1}_{\cal T}\ket{\pi_\ell^1}+\ket{0}_{\cal F}\ket{\ell}_{\cal C}\ket{0}_{\cal T}\ket{v_\ell} +\lambda_\ell\ket{1}_{\cal F}\ket{\ell}_{\cal C}\ket{0}_{\cal T}\ket{\pi_M}_{\cal V}
\end{align*}
so we conclude
\begin{equation}\label{eq:S_AA-action-1}
\ket{0}\ket{\ell}\ket{0}(\ket{\pi_{\ell}}+\ket{v_{\ell}})
     \overset{S_{\rm AA}(A)}{\longmapsto}
     \ket{0}\ket{\ell+1}\ket{1}\ket{\pi_\ell^1}+\ket{0}\ket{\ell}\ket{0}\ket{v_\ell} +\lambda_\ell\ket{1}\ket{\ell}\ket{0}\ket{\pi_M}.
\end{equation}

For the second case, we consider $t\in\{1,2,3\}$. Then for the $t$ part of the state, step 1 does nothing, since $t\neq 0$, and step 4 does nothing, since steps 2 and 3 only control on $t$, so it's still $\neq 0$. The action of steps 2 and 3 is to apply $G_t$ to the last register controlled on the flag being 0, and Step 5 increments ${\cal T}$, so that the action of $S_{\rm AA}(A)$ on the $t$ part of the state is:
\begin{equation}\label{eq:S_AA-action-2}
    \begin{split}
        \ket{0}\ket{\ell+1}\ket{t}\ket{\pi_{\ell}^t} \overset{2\& 3}{\longmapsto}{}& \ket{0}\ket{\ell+1}\ket{t}G_t\ket{\pi_{\ell}^t}\\
        ={}&\ket{0}\ket{\ell+1}\ket{t}\ket{\pi_{\ell}^{t+1}}\\
        \overset{5}{\longmapsto}{}&\ket{0}\ket{\ell+1}\ket{t+1}\ket{\pi_{\ell}^{t+1}}.
    \end{split}
\end{equation}
Finally, we consider the action on the $t=4$ part of the state. Again, Step 1 does nothing, since $t\neq 0$, and Step 2 does nothing because $t\not\in\{1,3\}$. Step 3 applies $G_4=S_F$ conditioned on the flag being 0. Step 4 does nothing since $t\neq 0$, and Step 5 increments ${\cal T}$ mod 5 in the ${\cal V}$ part of the state, so the action of $S_{\rm AA}(A)$ on the $t=4$ part of the state is:
\begin{align*}
    \ket{0}\ket{\ell+1}\ket{4}(\ket{\pi_{\ell}^4}+\ket{v_{\ell}'})
    \overset{3}{\mapsto}{}&
    \ket{0}\ket{\ell+1}\ket{4}S_F(\ket{\pi_{\ell}^4}+\ket{v_{\ell}'})\\
    ={}& \ket{0}\ket{\ell+1}\ket{4}\ket{\pi_{\ell+1}}+\ket{0}\ket{\ell+1}\ket{4}\ket{v_{\ell}'} & \mbox{by \eqref{eq:S_F-cat2}}\\
    \overset{5}{\mapsto}{}&
    \ket{0}\ket{\ell+1}\ket{0}\ket{\pi_{\ell+1}}+\ket{0}\ket{\ell+1}\ket{4}\ket{v_{\ell}'}.
\end{align*}
Combining this with \eqref{eq:S_AA-action-1} and \eqref{eq:S_AA-action-2} as well as \eqref{eq:coupling}, the claimed action of $S_{\rm AA}(A)$ follows.
\end{proof}

Now, define a catalyst, based on the vectors $\ket{v^\ell}$ from \Cref{lem:G-catalyst}:
\begin{equation}\label{eq:v-prime}
\ket{v'}=\ket{0}_{\cal F}\sum_{\ell=1}^\infty\ket{\ell}_{\cal C}\sum_{t=1}^5\ket{t}_{\cal T}\ket{\pi_{\ell-1}^t}_{\cal V}+\sum_{\ell=0}^\infty\ket{v^\ell}_{{\cal FCTL}_F},\qquad\mbox{where }\ket{5}_{\cal T}\ket{\pi^5_{\ell-1}}\equiv\ket0_{\cal T}\ket{\pi_\ell}.
\end{equation}
First, we note that the norm converges.
\begin{lemma}\label{lem:S_AA-canonical-complexity}
    $\norm{\ket{v'}}^2=O({(1+W(S_F^\dagger,{\cal K})})/\sqrt{\mu})$, and $\norm{\Pi_{{\cal L}^\bullet}\ket{v'}}^2=O(1/\sqrt{\mu})$, where ${\cal L}^\bullet$ is the query part of the private space, defined in \Cref{fig:amp-amp-transducer-2}.
\end{lemma}
\begin{proof}
    Note that for $t\in\{1,\dots,5\}$,
    $\norm{\ket{\pi_{\ell-1}^t}}^2 = \norm{\ket{\pi_{\ell}}}^2,$ {since $\ket{\pi_{\ell-1}^1},\dots,\ket{\pi_{\ell-1}^5}=\ket{\pi_\ell}$ are obtained from one another by the unitaries $A^\dagger$, $I-2\ketbra00$, $A$ and $U_F$.}
    Thus, since the {$\ket{v^\ell}$} are pairwise orthogonal and orthogonal to the first part of $\ket{v'}$:
    \begin{align*}
        \norm{\ket{v'}}^2&=\sum_{\ell=1}^\infty 5\norm{\ket{\pi_{\ell}}}^2 + \sum_{\ell=0}^\infty\norm{\ket{v^\ell}}^2
        {\leq} \sum_{\ell=1}^\infty 5\norm{\ket{\pi_{\ell}}}^2 + 2W(S_F^\dagger,{\cal K})\sum_{\ell=0}^\infty\norm{\ket{\pi_{\ell+1}}}^2\\
        &= (5+2W(S_F^\dagger,{\cal K}))\sum_{\ell=1}^\infty\norm{\ket{\pi_{\ell}}}^2.
    \end{align*}
    By \Cref{lem:S_AA-complexity}, this is $O((1+W(S_F^\dagger,{\cal K}))/\sqrt{\mu})$.
    To upper bound the size of the projection onto ${\cal L}^\bullet$, note that each $\ket{v^\ell}$ is supported on states with $\ket{0}$ or $\ket{4}$ in ${\cal T}$, so it is orthogonal to ${\cal L}^\bullet$. Thus,
        \begin{align*}
        \norm{\Pi_{{\cal L}^\bullet}\ket{v'}}^2{\leq} \sum_{\ell=1}^\infty 5\norm{\ket{\pi_{\ell}}}^2=O(1/\sqrt{\mu}),
    \end{align*}
which completes the proof.
\end{proof}

We next show that $\ket{v'}$ is indeed a catalyst for the desired initial state.

\begin{lemma}\label{lem:S_AA-canonical-action} Let $\ket{\lambda}:=\sum_{\ell=0}^\infty\lambda_\ell\ket{\ell}$, as in \Cref{lem:S_AA-action}. Then
    $$S_{\rm AA}(A)(\ket{0}_{\cal F}\ket{0}_{\cal C}\ket{0}_{\cal T}\ket{\pi}_{\cal V}+\ket{v'})=\ket{1}_{\cal F}\ket{\lambda}_{\cal C}\ket{0}_{\cal T}\ket{\pi_M}_{\cal V}+\ket{v'}.$$
    Thus, by \Cref{lem:S_AA-canonical-complexity},
    $$L(S_{\rm AA}(A),A,\ket0\ket0\ket0\ket\pi)=O(1/\sqrt\mu)
    \;\mbox{ and }\;
    W(S_{\rm AA}(A),\ket0\ket0\ket0\ket\pi) =O((1+W(S_F^\dagger,{\cal K}))/\sqrt\mu).$$
\end{lemma}
\begin{proof}
First note that since
$$\sum_{t=1}^5\ket{t}\ket{\pi_{\ell-1}^t}=\sum_{t=1}^4\ket{t}\ket{\pi_{\ell-1}^t}+\ket{0}\ket{\pi_\ell},$$
and $\ket{\pi}=\ket{\pi_0}$, we have:
\begin{align*}
    \ket{0}_{\cal F}\ket{0}_{\cal C}\ket{0}_{\cal T}\ket{\pi}_{\cal V}+\ket{v'}
    &=\ket{0}_{\cal F}\ket{0}_{\cal C}\ket{0}_{\cal T}\ket{\pi}_{\cal V}+\ket{0}_{\cal F}\sum_{\ell=1}^\infty\ket{\ell}_{\cal C}\sum_{t=1}^5\ket{t}_{\cal T}\ket{\pi_{\ell-1}^t}_{\cal V}+\sum_{\ell=0}^\infty\ket{v^\ell}_{{\cal FCTL}_F}\\
    &= \ket{0}\sum_{\ell=0}^\infty\ket{\ell}\ket{0}\ket{\pi_\ell}
    +\ket{0}\sum_{\ell=1}^\infty\ket{\ell}\sum_{t=1}^4\ket{t}\ket{\pi_{\ell-1}^t}+\sum_{\ell=0}^\infty\ket{v^\ell}\\
    &=\sum_{\ell=0}^\infty\left(\ket{0}\ket{\ell}\ket{0}\ket{\pi_\ell}+\ket{0}\ket{\ell+1}\sum_{t=1}^4\ket{t}\ket{\pi_\ell^t}+\ket{v^\ell}\right).
\end{align*}
Thus, by \Cref{lem:G-catalyst}, we have:
\begin{align*}
    S_{\rm AA}(A)\left(\ket{0}\ket{0}\ket{0}\ket{\pi}+\ket{v'}\right) &= \sum_{\ell=0}^\infty S_{\rm AA}(A)\left(\ket{0}\ket{\ell}\ket{0}\ket{\pi_\ell}+\ket{0}\ket{\ell+1}\sum_{t=1}^4\ket{t}\ket{\pi_\ell^t}+\ket{v^\ell}\right)\\
    &= \sum_{\ell=0}^\infty \left(\ket{0}\ket{\ell+1}\sum_{t=1}^5\ket{t}\ket{\pi_\ell^t}+\lambda_\ell\ket{1}\ket{\ell}\ket{0}\ket{\pi_M}+\ket{v^\ell}\right)\\
    &= \ket{1}\underbrace{\sum_{\ell=0}^\infty\lambda_\ell\ket{\ell}}_{=\ket{\lambda}}\ket{0}\ket{\pi_M}+\underbrace{\sum_{\ell=1}^\infty\ket{0}\ket{\ell}\sum_{t=1}^5\ket{t}\ket{{\pi_{\ell-1}^t}}+\sum_{\ell=0}^\infty\ket{v^\ell}}_{=\ket{v'}}.
\end{align*}
Note that exchanging $S_{\rm AA}(A)$ with the infinite sum is valid, since the sum converges, and $S_{\rm AA}(A)$ has bounded (unit) norm.
\end{proof}

We now show that we can truncate $S_{\rm AA}(A)$ to something that is finite and efficiently implementable, establishing \Cref{thm:aa-canonical}, which we restate here. 

\begin{corollary}\label{thm:aa-canonical-truncated}
For any integer $D\geq1$, let $S_{\rm AA}^{[D]}(A)$ be obtained from $S_{\rm AA}(A)$ in \Cref{fig:amp-amp-transducer-2} by replacing ${\cal C}$ with ${\cal C}^{[D]}$ and the incrementation in Step~4 by incrementation modulo $D$ (restricting the public and private spaces accordingly). Let $T<160/\sqrt\mu$ be as in the proof of \cref{lem:S_AA-complexity}. Then $S_{\rm AA}^{[D]}(A)$ is canonical with respect to the oracle that applies $A$ or $A^\dagger$, and:
\begin{enumerate}
    \item it transduces $\ket0_{\cal F}\ket0_{{\cal C}^{[D]}}\ket0_{\cal T}\ket\pi$ into a state $\ket{\tau}$ with $\norm{\ket\tau-\ket1_{\cal F}\ket{\lambda'}_{{\cal C}^{[D]}}\ket0_{\cal T}\ket{\pi_M}}\leq16T/D$, where $\ket{\lambda'}$ is as in \cref{thm:aa-truncated};
    \item $W(S_{\rm AA}^{[D]}(A),\ket0\ket0\ket0\ket\pi)=O((1+W(S_F^\dagger,{\cal K}))/\sqrt\mu)$ and $L(S_{\rm AA}^{[D]},A,\ket0\ket0\ket0\ket\pi)=O(1/\sqrt\mu)$;
    \item its work unitary uses one controlled call to each of $S_F$ and $S_F^\dagger$, $O(1)$ controlled calls to $R$, and $O(\log D+\log\dim{\cal V})$ additional one- and two-qubit gates.
\end{enumerate}
\end{corollary}
\begin{proof}
Item~3 follows as in the proof of \cref{thm:aa-truncated}: Step~1 implements $W_{\rm leak}(a_\ell)$ in $O(\log D)$ gates, with the control on being in the marked space computed using calls to $R$;
Step~2 uses controlled calls to $A$ and $A^\dagger$ (one each) with $O(1)$ cost to compute the control bits;
Steps~4 and~5 are increments modulo $D$ and~$5$ with $O(1)$ extra controls; and for Step~3: $G_0$ and $G_4$ use calls to $R$, $S_F$, $S_F^\dagger$ (one each), and $G_2=(1-2\ket{0}\bra{0})_{\cal V}\oplus I_{{\cal L}_F}$ costs $O(\log\dim{\cal V})$ gates. 

For the catalyst, truncate $\ket{v'}$ from \Cref{eq:v-prime} to the terms with $\ell\leq D-1$, reading the counter value $\ell+1$ modulo $D$:
$$\ket{v'^{[D]}}\defeq\sum_{\ell=0}^{D-1}\left(\ket{0}_{\cal F}\ket{\ell+1\bmod D}\sum_{t=1}^4\ket{t}_{\cal T}\ket{\pi_\ell^t}+\ket{v^\ell}\right)+\sum_{\ell=1}^{D-1}\ket{0}_{\cal F}\ket{\ell}\ket{0}_{\cal T}\ket{\pi_\ell},$$
where in $\ket{v^{D-1}}=\ket{0}\ket{D-1}\ket{0}\ket{v_{D-1}}+\ket{0}\ket{D}\ket{4}\ket{v_{D-1}'}$ (see \Cref{eq:v-sup-ell}) the counter value $D$ is also read as $0$. All of these terms lie in the private space: the only terms with ${\cal C}^{[D]}=0$ in the public space come from $\ell=D-1$ and have ${\cal T}\neq0$. The terms are also pairwise orthogonal, as in \cref{lem:G-catalyst}. The proof of \cref{lem:G-catalyst} is local in $\ell$, so it applies verbatim for each $\ell\leq D-1$. Summing as in the proof of \cref{lem:S_AA-canonical-action}, the only difference is the wrapped-around term $\ket0_{\cal F}\ket0_{\cal C}\ket0_{\cal T}\ket{\pi_D}$, which lies in the public space. So $S_{\rm AA}^{[D]}(A)$ transduces the input into $\ket\tau\defeq\ket1\ket{\lambda^{[D]}}\ket0\ket{\pi_M}+\ket0\ket0\ket0\ket{\pi_D}$, where $\ket{\lambda^{[D]}}$ is as in the proof of \Cref{thm:aa-truncated}, and Item~1 follows exactly as in the proof of \Cref{thm:aa-truncated}. 

Finally, since the terms of $\ket{v'^{[D]}}$ are (relabelings of) a subset of the pairwise orthogonal terms of $\ket{v'}$, we have $\norm{\ket{v'^{[D]}}}\leq\norm{\ket{v'}}$ and $\norm{\Pi_{{\cal L}^\bullet}\ket{v'^{[D]}}}\leq\norm{\Pi_{{\cal L}^\bullet}\ket{v'}}$, so Item~2 follows from \cref{lem:S_AA-canonical-complexity}.
\end{proof}

\subsection{An optimal transducer and algorithm for ground-state preparation}

We now have all the ingredients to define a transducer for ground-state preparation, from which we obtain an algorithm. 
It is relatively straightforward to construct a transducer that uses infinite space, but has the correct transduction complexity, and gives a query-optimal algorithm: we just compose the error-free transducers for filtering and amplitude amplification.
However, we want to turn this into an efficient algorithm, with small overhead, so we work with the truncated transducers and keep an error budget.

Let $S_F$ be the LCU transducer from \Cref{thm:finite-lcu-transducer}, with an error parameter $\eta_2>0$ that we fix in the proof of \cref{lem:composed-transducer} below. Denote by $U_F$ its transduction action, which acts as
$$U_F\ket{\psi_k}\ket{0}=\tilde f_\delta(E_k)\ket{\psi_k}\ket{0}+\ket{\psi_k^\bot}$$
for some $\ket{\psi_k^\bot}$ orthogonal to $\ket{0}$ in the second register.
Here and below, the last three registers are the system register $\C^N$ and the registers $\C^{2J+1}$ and $\C^J$ of $S_F$, so that $U_F$ acts on $\C^N\ot\C^{2J+1}$, and ${\cal V}=\C^N\ot\C^{2J+1}\ot\mathrm{span}\{\ket{0}\}$ plays the role of the public space of $S_F$ in \cref{thm:aa-canonical}.
When we plug the LCU transducer of \cref{thm:finite-lcu-transducer} into \cref{thm:aa-canonical}, 
we get the following result.

\begin{lemma}\label{lem:composed-transducer}
Fix $\eta\in (0,1)$. Let $(U,A)$ represent an instance of guided ground-state preparation with estimate.
    There is a transducer $S(A)$, canonical in the oracle that applies $A$ or $A^\dagger$, with transduction action
    \begin{align*}
        \ket{\phi_{\rm init}}\defeq\ket{0}\ket{0}\ket{0}U_F(A\ket{0}\ket{0})\ket{0}\overset{S(A)}{\rightsquigarrow} \ket{\tau} \qquad \text{ where } \quad \norm{\ket{\tau} - \ket{1}\ket{\lambda}\ket{0}\ket{\psi_0}\ket{0}\ket{0}} \leq \eta,
    \end{align*}
    where $\ket{\lambda}$ is some unit vector, such that
    \begin{enumerate}
        \item\label{it:W} The transduction complexity is $W(S(A),\ket{\phi_{\rm init}})=O(1/(\Delta\gamma))$.
        \item\label{it:LV} The Las Vegas query complexity is $L(S(A),A,\ket{\phi_{\rm init}})=O(1/\gamma)$.
        \item\label{it:call-U} The work unitary $S^\circ$ can be implemented using $O(1)$ controlled calls to each of $U$ and $U^\dagger$.
        \item\label{it:gates} Additionally, one application of the work unitary uses $O\left(\log N + (\log \frac{1}{\delta \eta \gamma})^2\right)$ one- and two-qubit gates and $O(\log \frac{1}{\delta \eta \gamma})$ additional qubits.
    \end{enumerate}
\end{lemma}

\begin{proof}
Recall that we assume we have an estimate $\tilde{E}_0$ of $E_0$ to precision $\delta = \Delta/3$, which implies there is a gap of at least $2\delta$ between $\tilde{E}_0$ and $E_1$.
    By replacing $U$ with $e^{-i\tilde{E}_0}U$ -- equivalently, $H$ with $H-\tilde{E}_0I$ -- we may assume without loss of generality that $|E_0| \leq \delta$, so for the ideal filter we have $f_\delta(E_0) \geq \frac{1}{\pi}$, and we have $E_k \geq 2\delta$ (and, modulo $2\pi$, $E_k\leq2\pi-2\delta$, by our assumption $E_k\in (0,\pi]$ for $k\geq 1$) for $k \geq 1$, and hence $f_\delta(E_k) = 0$ (see \Cref{eq:f-delta}).
    Let $\ket\psi\defeq A\ket0=\sum_k\alpha_k\ket{\psi_k}$.
    We now consider the transducer $S(A)$ from \cref{thm:aa-canonical} with error parameter $\eta_1 = \eta/2$, calling $S_F$ from \cref{thm:finite-lcu-transducer} with error parameter $\eta_2 = \eta \gamma / (4\pi)$, and with marked subspace consisting of the states with $\ket0$ in the second register.
    Then we have:
    $$\ket{\pi}=U_F(A\ket{0}\ket{0})\ket{0} = \bigl(\tilde f_\delta(H)\ket{\psi}\ket{0}+\ket{\psi^\bot}\bigr)\ket{0}
    \quad\mbox{and}\quad
    \ket{\pi_M}=\frac{{\tilde f}_\delta(H)\ket{\psi}}{\norm{{\tilde f}_\delta(H)\ket{\psi}}}\ket{0}\ket{0}.$$

    Denote by {$\mu = \norm{f_\delta(H) \ket{\psi}}^2$} and {$\tilde \mu = \norm{\tilde f_\delta(H) \ket{\psi}}^2$} the {weights} of the marked states for the ideal and truncated filters.
    Then:
    \begin{align*}
        \abs{\sqrt\mu - \sqrt{\tilde \mu}} \leq \norm{\tilde f_\delta(H) \ket{\psi} - f_\delta(H) \ket{\psi}} \leq \eta_2.
    \end{align*}
    In particular, the marked {weight} is {$\Omega(\gamma^2)$}: $\sqrt{\tilde\mu}\geq\sqrt\mu-\eta_2\geq\frac{\gamma}{\pi}-\frac{\eta\gamma}{4\pi}\geq\frac{3\gamma}{4\pi}$, using  {$\sqrt\mu$}$=\alpha_0f_\delta(E_0)\geq\gamma/\pi$.

    In order to apply \Cref{thm:aa-canonical}, we need an upper bound on $W(S_F^\dagger,{\cal K})$, the supremum of $W(S_F^\dagger,\ket{\varphi})$ over all unit vectors $\ket{\varphi}\in\mathrm{span}\{\ket{\pi_M},\ket{\pi_{\overline{M}}}\}$.
    \begin{claim}
        $W(S_F^\dagger,{\cal K})=O(1/\delta)$.
    \end{claim}
\begin{proof}
Let $X$ be the unitary that maps $\ket{j}\ket{k}\mapsto\ket{j}\ket{-k\bmod\abs{j}}$ on the last two registers for $0\leq k<\abs j$, and acts as the identity for $k\geq\abs j$ (in particular for $j=0$). Then $X$ acts as the identity on the public space and preserves the private space, commutes with $P$, and satisfies $X\,\mathsf{Sel}(U^\dagger)\,X^\dagger=\mathsf{Sel}(U)^\dagger$. Hence $S_F(U)^\dagger=X\,S_F(U^\dagger)\,X^\dagger$. So if $\ket v$ is a catalyst for a public vector $\ket\xi$ in $S_F(U^\dagger)$, then $X\ket v$ is a catalyst for $\ket\xi$ in $S_F(U)^\dagger$, of the same norm. Therefore applying \Cref{thm:finite-lcu-transducer} to $U^\dagger=e^{-iH}$ gives $W(S_F^\dagger,\ket{\varphi}\ket0\ket0)=O(1/\delta)$ for every unit $\ket{\varphi}\in\C^N$. In particular, this holds for $\ket{\pi_M}=\frac{\tilde f_\delta(H)\ket\psi}{\sqrt{\tilde\mu}}\ket0\ket0$. Moreover, if $\ket v$ is a catalyst for $A\ket0\ket0\ket0$ in $S_F$, which exists with $\norm{\ket v}^2=O(1/\delta)$ by \Cref{thm:finite-lcu-transducer}, then $S_F^\dagger(\ket\pi+\ket v)=A\ket0\ket0\ket0+\ket v$, so $\ket v$ is a catalyst for $\ket\pi$ in $S_F^\dagger$, and $W(S_F^\dagger,\ket\pi)=O(1/\delta)$. Now write a unit vector $\ket y\in{\cal K}$ as $\ket y=c_1\ket\pi+c_2\ket{\pi_M}$. Adding the corresponding catalysts gives $W(S_F^\dagger,\ket y)\leq2(\abs{c_1}^2+\abs{c_2}^2)\,O(1/\delta)=O(1/\delta)$.
\end{proof}

    We then conclude from \Cref{thm:aa-canonical} that $S(A)$ satisfies the claims in \Cref{it:W}, \Cref{it:LV}, and \Cref{it:call-U}: indeed, $W(S(A),\ket{\phi_{\rm init}})=O((1+W(S_F^\dagger,{\cal K}))/\sqrt{\tilde\mu})=O(1/(\delta\gamma))$, and $L(S(A),A,\ket{\phi_{\rm init}})=O(1/\sqrt{\tilde\mu})=O(1/\gamma)$.
    We need to verify that $\ket{\tau}$ is close to the ground state with garbage, as claimed.
    \cref{thm:aa-canonical} yields that there exists a unit vector $\ket{\lambda}$ such that
    \begin{align*}
        \Bigl\| \ket{\tau} - \tilde \mu^{-1/2} \ket{1}\ket{\lambda}\ket{0}\tilde f_\delta(H) \ket{\psi}\ket{0}\ket{0}\Bigr\| \leq \eta_1 = \frac{\eta}{2}.
    \end{align*}
    The claim follows once we prove that the normalized state $\frac{\tilde f_\delta(H) \ket{\psi}}{\sqrt{\tilde \mu}}$ is $\eta/2$-close in norm to $\ket{\psi_0}$.
    To this end we first recall that $\ket{\psi_0} = \frac{f_\delta(H) \ket{\psi}}{\sqrt\mu}$, since $f_\delta(E_k)=0$ for $k\geq1$ and $\alpha_0f_\delta(E_0)>0$. We then find
    \begin{align*}
        \Bigl\|\frac{f_\delta(H) \ket{\psi}}{\sqrt\mu} - \frac{\tilde f_\delta(H) \ket{\psi}}{\sqrt{\tilde \mu}}\Bigr\| \leq \frac{1}{\sqrt\mu} \norm{f_\delta(H) \ket{\psi} - \tilde f_\delta(H) \ket{\psi}} + \frac{1}{\sqrt\mu} \abs{\sqrt\mu - \sqrt{\tilde \mu}} \leq \frac{2\pi\eta_2}{\gamma} \leq \eta/2.
    \end{align*}
    
    Finally, we need to account for the gate and space overheads in \Cref{it:gates}.
    Note that with our choice of $\eta_1$ and $\eta_2$, in \Cref{thm:finite-lcu-transducer} we have $J = \Theta(\delta^{-3/2}\eta^{-1} \gamma^{-1})$, and in \Cref{thm:aa-canonical} this yields $\log\dim \mathcal V = O(\log N + \log J)$.
    The work unitary can be implemented using $O(1)$ calls to $S_F$ and $S_F^\dagger$, which in turn need a single controlled call to each of $U$ and $U^\dagger$.
    Additionally, $S_F$ needs $O((\log J)^2) = O((\log \frac{1}{\delta \eta \gamma})^2)$ one- and two-qubit gates.
    From \cref{thm:aa-canonical} the work unitary needs a further $O(\log \frac{1}{\eta \gamma} + \log \dim \mathcal V) = O(\log \frac{N}{\eta \gamma} + \log \frac{1}{\delta \eta \gamma})$ one- and two-qubit gates (this includes implementing the reflection $R=2\Pi_M-I$, which acts on $O(\log J)$ qubits), so in total we need $O(\log \frac{N}{\eta \gamma} + (\log \frac{1}{\delta \eta \gamma})^2)$ one- and two-qubit gates.
    The space overhead is $O(\log J + \log D) = O(\log \frac{1}{\delta \eta \gamma})$, using $D=O(1/(\eta_1\sqrt{\tilde\mu}))=O(1/(\eta\gamma))$ from \cref{thm:aa-canonical}.
\end{proof}

We can now apply \Cref{thm:transducer-to-alg} to prove our main claim.

\begin{theorem}\label{thm:gsp-constant}
Let $(U,A)$ be an instance of guided ground-state preparation with estimate, with parameters $\Delta$ and $\gamma$. There is a quantum algorithm that outputs a state $\ket{\tilde\tau}$ such that
$$\norm{\ket{\tilde\tau}-\ket{\psi_0}\ket{g}}\leq\frac13$$
for some unit vector $\ket{g}$ on the auxiliary registers (where we order the registers so that the system register $\C^N$ comes first), and that uses
\begin{enumerate}
    \item $O(1/\gamma)$ controlled calls to $A$ and $A^\dagger$;
    \item $O(1/(\gamma\Delta))$ controlled calls to $U$ and $U^\dagger$;
    \item $O\!\left(\frac{1}{\gamma\Delta}\left(\log N+\log^2\frac{1}{\gamma\Delta}\right)\right)$ additional one- and two-qubit gates;
    \item $O(\log\frac{1}{\gamma\Delta})$ auxiliary qubits, in addition to the $\ceil{\log N}$ system qubits.
\end{enumerate}
\end{theorem}

\begin{proof}
Fix constants $\eps_0=\eps_1=\eta=\frac19$, and let $\delta=\Delta/3$. Let $S_F$ be the truncated LCU transducer from \Cref{thm:finite-lcu-transducer}, with the parameters chosen in the proof of \cref{lem:composed-transducer}, and let $U_F$ be its transduction action. With these parameters, $J=\Theta(\delta^{-3/2}\gamma^{-1})$ and $D=O(1/\gamma)$, so $\log J,\log D=O(\log\frac{1}{\gamma\Delta})$. The algorithm consists of two phases. We suppress the auxiliary registers of the algorithms from \Cref{thm:transducer-to-alg}; each of these is returned to $\ket{0}$ up to the stated error.

\paragraph{Phase 1: initial state preparation.} Prepare $A\ket{0}\ket{0}\ket{0}$, where the last two registers are the LCU registers $\C^{2J+1}\otimes\C^J$ of $S_F$. Let $V_1$ be the algorithm obtained from applying \Cref{thm:transducer-to-alg} to the transducer $S_F$. Since $S_F$ makes no calls to $A$, we may regard it as canonical with a trivial oracle, and take $L=0$. By \Cref{thm:finite-lcu-transducer}, $W(S_F,A\ket0\ket0\ket0)=O(1/\delta)$, so we may take $W=O(1/\delta)$. Then
$$\norm{V_1A\ket{0}\ket{0}\ket{0}-U_F(A\ket{0}\ket0)\ket0}\leq\eps_0.$$
By using the same truncated filter $S_F$ as the one composed into $S(A)$ in \cref{lem:composed-transducer}, the ideal output of this phase is exactly the input $\ket{\phi_{\rm init}}$ of Phase~2 (after appending the registers $\ket0_{\cal F}\ket0_{{\cal C}^{[D]}}\ket0_{\cal T}$).

\paragraph{Phase 2: amplitude amplification.} Let $V_2$ be the algorithm obtained from applying \Cref{thm:transducer-to-alg} to the canonical transducer $S(A)$ of \cref{lem:composed-transducer}, with $W=O(1/(\gamma\Delta))$ and $L=O(1/\gamma)$. Both bounds are witnessed by a single catalyst (the truncation of $\ket{v'}$ from \Cref{eq:v-prime}), as \Cref{thm:transducer-to-alg} requires. Since $\ket{\phi_{\rm init}}$ lies in the public space of $S(A)$, we get
$$\norm{V_2\ket{\phi_{\rm init}}-\ket\tau}\leq\eps_1,$$
where $\ket\tau$ is the transduction action from \cref{lem:composed-transducer}.

\paragraph{Correctness.} Let $\ket{\tilde\tau}\defeq V_2V_1A\ket0\ket0\ket0$ (with ${\cal F},{\cal C}^{[D]},{\cal T}$ initialised to $\ket0$), and let $\ket{g}$ collect the auxiliary part of $\ket1\ket\lambda\ket0\ket{\psi_0}\ket0$. Since $V_2$ is unitary, the triangle inequality and \cref{lem:composed-transducer} give
\begin{align*}
\norm{\ket{\tilde\tau}-\ket{\psi_0}\ket{g}}
&\leq\norm{V_2V_1A\ket0\ket0\ket0-V_2\ket{\phi_{\rm init}}}+\norm{V_2\ket{\phi_{\rm init}}-\ket\tau}+\norm{\ket\tau-\ket{\psi_0}\ket g}\\
&\leq\eps_0+\eps_1+\eta\leq\frac13.
\end{align*}

\paragraph{Complexity.} Throughout, a controlled call to a work unitary that itself makes controlled calls to $U$ or $A$ costs a controlled call to $U$ or $A$ plus $O(1)$ gates and one auxiliary qubit.
\begin{itemize}
    \item \emph{Calls to $A$ and $A^\dagger$.} Phase~1 uses one call to $A$. Phase~2 uses $O(L/\eps_1^2)=O(1/\gamma)$ controlled calls to $A$ and $A^\dagger$, so the total number of calls is $O(1/\gamma)$.
    \item \emph{Calls to $U$ and $U^\dagger$.} Phase~1 uses $O(1+W/\eps_0^2)=O(1/\delta)$ controlled calls to $S_F$, each of which makes one controlled call to each of $U$ and $U^\dagger$. Phase~2 uses $O(1+W/\eps_1^2)=O(1/(\gamma\delta))$ controlled calls to the work unitary $S^\circ$, each of which makes $O(1)$ controlled calls to $U$ and $U^\dagger$ by \cref{lem:composed-transducer}, Item~\ref{it:call-U}. The total is $O(1/(\gamma\delta))=O(1/(\gamma\Delta))$.
    \item \emph{Other gates.} By \Cref{thm:finite-lcu-transducer}, each call to $S_F$ in Phase~1 costs $O(\log^2 J)=O(\log^2\frac{1}{\gamma\Delta})$ gates, and \Cref{thm:transducer-to-alg} adds $O(1/\delta)$ gates. This gives $O(\frac1\Delta\log^2\frac{1}{\gamma\Delta})$ gates in total. By \cref{lem:composed-transducer}, Item~\ref{it:gates}, each call to $S^\circ$ in Phase~2 costs $O(\log N+\log^2\frac{1}{\gamma\Delta})$ gates, and \Cref{thm:transducer-to-alg} adds $O(1/(\gamma\delta))$ gates. This gives $O(\frac{1}{\gamma\Delta}(\log N+\log^2\frac{1}{\gamma\Delta}))$ gates, which dominates.
    \item \emph{Space.} The registers $\C^{2J+1}\otimes\C^J$, ${\cal F}$, ${\cal C}^{[D]}$ and ${\cal T}$ use $O(\log J+\log D)=O(\log\frac{1}{\gamma\Delta})$ qubits. By \Cref{thm:finite-lcu-transducer}, $S_F$ uses $O(\log J)$ further auxiliary qubits. The algorithms of \Cref{thm:transducer-to-alg} use $O(\log K)$ qubits for $K=O(1/(\gamma\Delta))$ applications.
\end{itemize}
This proves the theorem.
\end{proof}

The proof of \Cref{thm:main} now follows easily. \Cref{thm:gsp-constant} states it for constant error probability~$\eps$ and constant trace distance~$\eta$. To obtain arbitrary $\eps$ and $\eta$, we can apply the procedure from \Cref{ssec:HMdWalgoreduc} using the algorithm of \cref{thm:gsp-constant} as the initial algorithm~$\cal A$, with  the same bookkeeping as in \Cref{ssec:HMdWalgoreduc}.

\bibliographystyle{alpha}

\bibliography{bibo}

\begin{thebibliography}{HMdW03}

\bibitem[ARZ26]{apers2026elfs}
Simon Apers, J{\'e}r{\'e}mie Roland, and Yuxin Zhang.
\newblock Elfs, transducers and quantum walks, 2026.
\newblock arXiv:2605.30013.

\bibitem[BJ26]{belovs2024purifier}
Aleksandrs Belovs and Stacey Jeffery.
\newblock Space-efficient quantum error reduction without log factors.
\newblock {\em Quantum}, 10:2039, 2026.

\bibitem[BJY24]{belovs2024taming}
Aleksandrs Belovs, Stacey Jeffery, and Duyal Yolcu.
\newblock Taming quantum time complexity.
\newblock {\em Quantum}, 8:1444, 2024.

\bibitem[BY23]{belovs2023LasVegas}
Aleksandrs Belovs and Duyal Yolcu.
\newblock One-way ticket to {Las Vegas} and the quantum adversary.
\newblock arXiv:2301.02003, 2023.

\bibitem[CGWZ26]{chen2026optimal}
Boyang Chen, Minbo Gao, Xinzhao Wang, and Shuo Zhou.
\newblock Optimal query complexity for ground-state preparation, 2026.
\newblock arXiv:2609.35668.

\bibitem[HMdW03]{hmw:berrorsearch}
Peter H{\o}yer, Michele Mosca, and Ronald de~Wolf.
\newblock Quantum search on bounded-error inputs.
\newblock In {\em Proceedings of 30th International Colloquium on Automata, Languages and Programming (ICALP'03)}, volume 2719 of {\em LNCS}, pages 291--299. Springer, 2003.
\newblock quant-ph/0304052.

\bibitem[JW26]{JW:optQPE}
Stacey Jeffery and Freek Witteveen.
\newblock Optimal quantum algorithm for ground-state energy estimation with a guiding state, 2026.
\newblock arXiv:2608.24494.

\bibitem[LT20]{linlin&tong:groundstateprep}
{Lin Lin} and Yu~Tong.
\newblock Near-optimal ground state preparation.
\newblock {\em Quantum}, 4:372, 2020.
\newblock arXiv:2002.12508.

\bibitem[MdW26]{Mande2026tightboundsquantum}
Nikhil~S. Mande and Ronald de~Wolf.
\newblock Tight bounds for quantum phase estimation and related problems.
\newblock {\em Quantum}, 10:2140, 2026.
\newblock Earlier version in ESA'23. arXiv:2305.04908.

\bibitem[Miz09]{Mizel2008dampedSearch}
Ari Mizel.
\newblock Critically damped quantum search.
\newblock {\em Physical Review Letters}, 102:150501, 2009.

\bibitem[RIK22]{rendon2022effects}
Gumaro Rendon, Taku Izubuchi, and Yuta Kikuchi.
\newblock Effects of cosine tapering window on quantum phase estimation.
\newblock {\em Physical Review D}, 106(3):034503, 2022.

\bibitem[SdW26]{SdW:optQPE}
Rolando~D. Somma and Ronald de~Wolf.
\newblock Optimal lower bound for ground-state energy estimation with a guiding state, 2026.
\newblock arXiv:2608.24493.

\end{thebibliography}

\end{document}